\documentclass[sigplan,screen]{acmart}

\copyrightyear{2022}
\acmYear{2022}
\setcopyright{acmlicensed}\acmConference[EuroSys '22]{Seventeenth European
Conference on Computer Systems}{April 5--8, 2022}{RENNES, France}
\acmBooktitle{Seventeenth European Conference on Computer Systems (EuroSys
'22), April 5--8, 2022, RENNES, France}
\acmPrice{15.00}
\acmDOI{10.1145/3492321.3519594}
\acmISBN{978-1-4503-9162-7/22/04}

\usepackage{amsmath,amsfonts,amsthm}
\usepackage{algorithm}
\usepackage[noend]{algpseudocode}
\usepackage[shortlabels]{enumitem}
\usepackage{graphicx}
\usepackage{textcomp}
\usepackage{xcolor}
\usepackage{url}
\usepackage{xspace}
\usepackage{cleveref}
\usepackage{pifont}
\usepackage{booktabs}
\usepackage{wrapfig}
\usepackage[font=footnotesize]{caption}
\usepackage{thmtools}
\usepackage{thm-restate}
\usepackage{hyperref}
\usepackage{tabularx}
\usepackage{enumitem}

\algdef{SE}[Upon]{Upon}{EndUpon}[1]{\textbf{upon}\ #1\ \algorithmicdo}{\algorithmicend\ \textbf{}}%
\algtext*{EndUpon}
\newcommand{\cmark}{\ding{51}}%
\newcommand{\xmark}{\ding{55}}%

\microtypecontext{spacing=nonfrench}
\usepackage{todonotes}

\def\cameraReady{} 

\begin{document}


\ifdefined\cameraReady
\renewcommand{\todo}[2][]{}
\fi

\newcommand{\textcomment}[2]{\textbf{#1:} #2}
\newcommand\alberto[1]{\todo[color=yellow,inline]{\textcomment{Alberto}{#1}}}
\newcommand\sasha[1]{\todo[color=green,inline]{\textcomment{Sasha}{#1}}}
\newcommand\george[1]{\todo[color=brown,inline]{\textcomment{George}{#1}}}
\newcommand\oded[1]{\todo[color=orange,inline]{\textcomment{Oded}{#1}}}
\newcommand\lefteris[1]{\todo[color=orange,inline]{\textcomment{lefteris}{#1}}}
\newcommand\com[1]{}

\newcommand\sysname{Narwhal\xspace}
\newcommand\consname{Tusk\xspace}
\newcommand\mempool{Mempool\xspace}

\newcommand\chainspace{Chainspace\xspace}
\newcommand\ethereum{Ethereum\xspace}
\newcommand\hyperledger{Hyperledger\xspace}
\newcommand\omniledger{Omniledger\xspace}
\newcommand\rapidchain{RapidChain\xspace}
\newcommand\coconut{Coconut\xspace}
\newcommand\bft{BFT\xspace}
\newcommand\rscoin{RSCoin\xspace}
\newcommand\bftsmart{\textsc{bft-SMaRt}\xspace}
\newcommand\byzcuit{Byzcuit\xspace}
\newcommand\bitcoin{Bitcoin\xspace}
\newcommand\bitcoinng{Bitcoin-NG\xspace}
\newcommand\cosi{CoSi\xspace}
\newcommand\byzcoin{ByzCoin\xspace}
\newcommand\elastico{Elastico\xspace}
\newcommand\algorand{Algorand\xspace}
\newcommand\hyperledgerfabric{Hyperledger Fabric\xspace}
\newcommand\pbft{PBFT\xspace}
\newcommand\bftlong{Byzantine Fault-Tolerant\xspace} 
\newcommand\solidus{Solidus\xspace}
\newcommand\hashgraph{Hashgraph\xspace}
\newcommand\avalanche{Avalanche\xspace}
\newcommand\blockmania{Blockmania\xspace}
\newcommand\stellar{Stellar\xspace}
\newcommand\libra{Libra\xspace}
\newcommand\librabft{LibraBFT\xspace}
\newcommand\move{Move\xspace}
\newcommand\hotstuff{HotStuff\xspace}
\newcommand\vanillahs{Vanilla-HotStuff\xspace}
\newcommand\batchedhs{Batched-HotStuff\xspace}

\newcommand{\keyword}[1]{\normalfont \texttt{#1}}
\newcommand{\accounts}{\keyword{accounts}}

\newcommand{\transfer}{O}
\newcommand{\cert}{C}
\newcommand{\sync}{S}
\newcommand{\account}{a}
\newcommand{\authority}{\alpha}

\newcommand{\para}[1]{\vspace{2mm}\noindent\textbf{#1}.\xspace}

\newcommand{\cf}{cf.\@\xspace}
\newcommand{\vs}{vs.\@\xspace}
\newcommand{\etc}{etc.\@\xspace}
\newcommand{\ala}{ala\@\xspace}
\newcommand{\wrt}{w.r.t.\@\xspace}
\newcommand{\etal}{\textit{et al.}\@\xspace}
\newcommand{\eg}{\textit{e.g.}\@\xspace}
\newcommand{\ie}{\textit{i.e.}\@\xspace}
\newcommand{\via}{\textit{via}\@\xspace}
\newcommand{\defacto}{\textit{de facto}\@\xspace}

\newtheorem{assumption}{Security Assumption}

\newcommand\inlinesection[1]{{\bf #1.}}

\def\first{({i})\xspace}
\def\second{({ii})\xspace}
\def\third{({iii})\xspace}
\def\fourth{({iv})\xspace}
\def\fifth{({v})\xspace}
\def\sixth{({vi})\xspace}

\newcommand{\one}{({i})\xspace}
\newcommand{\two}{({ii})\xspace}
\newcommand{\three}{({iii})\xspace}
\newcommand{\four}{({iv})\xspace}
\newcommand{\five}{({v})\xspace}
\newcommand{\six}{({vi})\xspace}

\definecolor{verylightgray}{gray}{0.9}

\newcommand\vgap{\vskip 2ex}
\newcommand\marker{\vgap\ding{118}\xspace}
\def\na{--}
\def\unsure{?}
\def\missing{$!$}
\newcommand{\yes}{\ding{51}}
\newcommand{\no}{\ding{55}}
\DeclareRobustCommand\pie[1]{
\tikz[every node/.style={inner sep=0,outer sep=0, scale=1.5}]{
\node[minimum size=1.5ex] at (0,-1.5ex) {}; 
\draw[fill=white] (0,-1.5ex) circle (0.75ex); \draw[fill=black] (0.75ex,-1.5ex) arc (0:#1:0.75ex); 
}}
\def\L{\pie{0}} 
\def\M{\pie{-180}} 
\def\H{\pie{360}} 


\algdef{SE}[Receiving]{Receiving}{EndReceiving}[1]{\textbf{upon
receiving}\ #1\ \algorithmicdo}{\algorithmicend\ \textbf{}}%
\algtext*{EndReceiving}
\newcommand\StateX{\Statex\hspace{\algorithmicindent}}
\algrenewcommand\textproc{}

\graphicspath{{figures/}}

\date{}

\title{
\sysname and \consname: A DAG-based Mempool and Efficient BFT Consensus
}

\ifdefined\cameraReady

\author{George Danezis}
\affiliation{\institution{Mysten Labs \& UCL}}

\author{Lefteris Kokoris-Kogias} 
\affiliation{\institution{IST Austria}}

\author{Alberto Sonnino}
\affiliation{\institution{Mysten Labs}}

\author{Alexander Spiegelman}
\affiliation{\institution{Aptos}}

\else
\author{}
\fi

\settopmatter{printfolios=false}
\settopmatter{printacmref=true}

\begin{abstract}

We propose separating the task of reliable transaction dissemination from transaction ordering, to enable
high-performance Byzantine fault-tolerant quorum-based consensus. We design and evaluate a mempool protocol, \sysname, specializing in high-throughput reliable dissemination and storage of causal histories of transactions. \sysname tolerates an asynchronous network and maintains high performance despite failures. 
\sysname is designed to easily scale-out using multiple workers at each validator, and we demonstrate that there is no foreseeable limit to the throughput we can achieve. 

Composing \sysname with a partially synchronous consensus protocol (\sysname-HotStuff) yields significantly better throughput even in the presence of faults or intermittent loss of liveness due to asynchrony.
However, loss of liveness can result in higher latency. To achieve overall good performance when faults occur we design \consname, a zero-message overhead asynchronous consensus protocol, to work with \sysname. We demonstrate its high performance under a variety of configurations and faults.

As a summary of results, on a WAN,  \sysname-Hotstuff achieves over 130,000 tx/sec at less than  2-sec latency compared with 1,800 tx/sec at 1-sec latency for Hotstuff. Additional workers increase throughput linearly to 600,000 tx/sec without any latency increase. \consname achieves 160,000 tx/sec with about 3 seconds latency. Under faults, both protocols maintain high throughput, but \sysname-HotStuff suffers from increased latency.

\end{abstract}



\maketitle

\section{Introduction} \label{sec:introduction}



Byzantine consensus protocols~\cite{dolev1982efficient,bracha1987asynchronous,CastroL02} and the state machine replication paradigm~\cite{bessani2014state} for building reliable distributed systems have been studied for over 40 years. However, with the rise in popularity of blockchains there has been a renewed interest in engineering high-performance consensus protocols. Specifically, to improve on Bitcoin's~\cite{nakamoto2008bitcoin} throughput of only 4 tx/sec early works~\cite{DBLP:journals/corr/Kokoris-KogiasJ16} suggested committee based consensus protocols. For higher throughput and lower latency committee-based protocols are required, and are now becoming the norm in proof-of-stake designs.


Existing approaches to increasing the performance of distributed ledgers focus on creating lower-cost consensus algorithms culminating with Hotstuff~\cite{DBLP:conf/podc/YinMRGA19}, which achieves linear message complexity in the partially synchronous setting. To achieve this, Hotstuff leverages a leader who collects, aggregates, and broadcasts the messages of other validators.
However, theoretical message complexity should not be the only optimization target. More specifically:
\begin{itemize}[noitemsep,topsep=0pt]
\item Any (partially-synchronous) protocol that minimizes overall message number, but relies on a leader to produce proposals and coordinate consensus, fails to capture the high load this imposes on the leader who inevitably becomes a bottleneck.

\item Message complexity counts the number of \emph{metadata} messages (e.g., votes, signatures, hashes) which take minimal bandwidth compared to the dissemination of bulk transaction data (blocks). Since blocks are orders of magnitude larger (10MB) than a typical consensus message (100B), the asymptotic message complexity is practically amortized for fixed mid-size committees (up to $\sim50$ nodes).
\end{itemize}

Additionally, consensus protocols have grouped a lot of functions into a monolithic protocol.
In a typical distributed ledger, such as Bitcoin or LibraBFT\footnote{LibraBFT recently renamed to DiemBFT.}~\cite{baudet2019state}, clients send transactions to a validator that shares them using a \mempool protocol. Then a subset of these transactions are periodically re-shared and committed as part of the consensus protocol.
Most research so far aims to increase the throughput of the consensus layer.

This paper formulates the following hypothesis: \textbf{a better \mempool, that reliably distributes transactions, is the key enabler of a high-performance ledger. It should be separated from the consensus protocol altogether, leaving consensus only the job of ordering small fixed-size references. This leads to an overall system throughput being largely unaffected by consensus throughput.}

This work confirms the hypothesis; monolithic protocols place transaction dissemination in the critical path of consensus, 
impacting performance more severely than consensus itself. With \sysname, we show that we can off-load \emph{reliable} transaction dissemination to the \mempool protocol, and only rely on consensus to sequence a very small amount of metadata, increasing performance significantly.
Therefore, there is a clear gap between what in theory is asymptotically optimal for an isolated consensus protocol and what offers good performance in a real distributed ledger.

Prior work into closing this gap both in permissionless~\cite{graphene, prism} and permissioned protocols~\cite{CastroL02} focuses on separating the transmission of large messages and metadata but not on guaranteeing reliability. 
As a result, they work exceptionally well under no faults but suffer severely at the slightest network failure.
In order to evaluate this existing approach, we adapt Hotstuff to separate block dissemination into a separate \mempool layer and call the resulting system Batched-HS (\Cref{sec:implementation}). In Batched-HS, validators broadcast blocks of transactions in a \mempool and the leader proposes block hashes during consensus, instead of broadcasting transactions in the critical path, gaining up to 50x in performance over the existing implementation.
This design, however, only performs well under ideal network conditions where the proposal of the leader is available to the majority of the validators promptly.
To make a robust \mempool we design \sysname, a DAG-based, structured \mempool which implements causal order reliable broadcast of transaction blocks, exploiting the available resources of validators in full. Combining the \sysname \mempool with HotStuff (\sysname-HS) provides good throughput even under faults or unstable network conditions (but at an inevitable higher latency).
%
To reduce latency under faults and asynchrony, we can extend \sysname with a random coin to provide asynchronous consensus, which we call \consname.
\consname is a fully-asynchronous, wait-free consensus where each party decides the agreed values by examining its local DAG \textit{without sending any additional messages}. 

\begin{figure}[t]
\vspace{-0.3cm}
    \centering
    \includegraphics[width=\columnwidth]{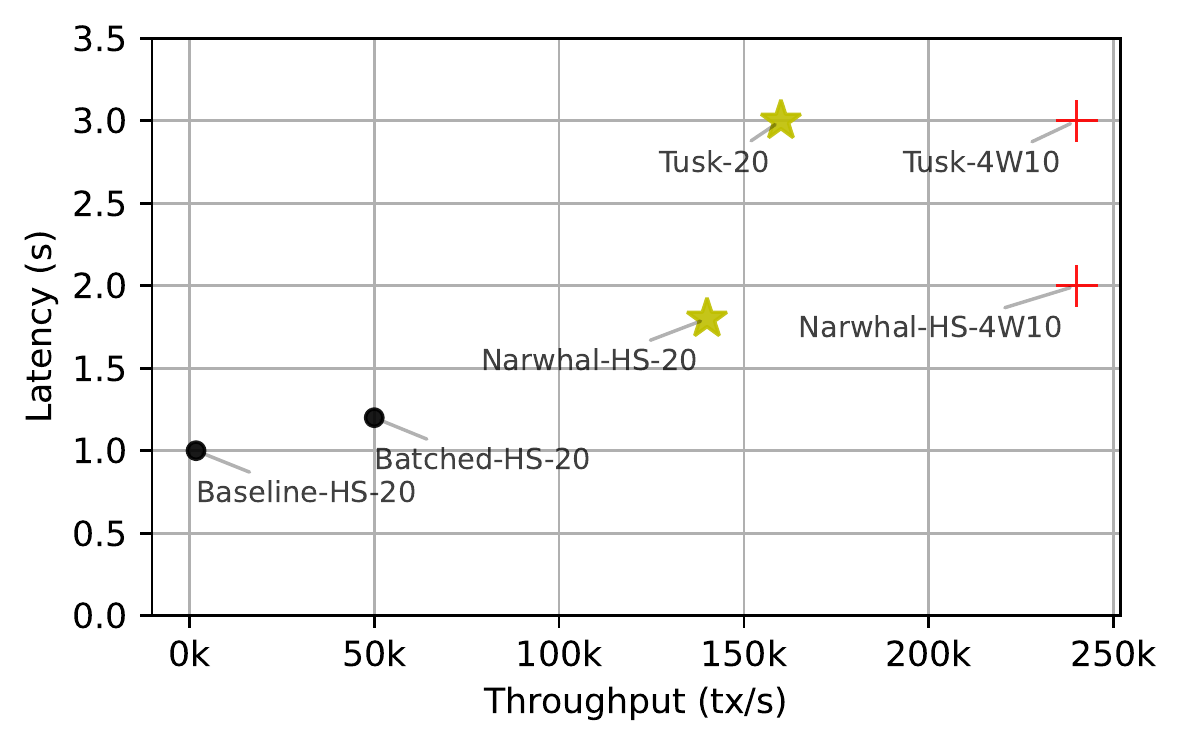}
    \caption{Summary of WAN performance results, for consensus systems with traditional mempool (circle), \sysname mempool (star), and many workers (cross). Transactions are 512B.}
            \vspace{-0.4cm}
    \label{fig:summary}
\end{figure}

\para{Contributions}
We make the following contributions:
\begin{itemize}[noitemsep,topsep=0pt]
    \item We build \sysname, an advanced \mempool protocol that guarantees optimal throughput (based on network speed) even under asynchrony and combines it with our Hotstuff implementation to see increased throughput at a modest expense of latency.
    \item We leverage the structure of \sysname and enhance it with randomness to get Tusk, a practical extension of DAG-Rider~\cite{DBLP:journals/corr/abs-2102-08325}. Tusk is a high-throughput, DDoS resilient, and zero overhead consensus protocol. We demonstrate experimentally its high performance in a WAN, even when faults occur.
\end{itemize}

\Cref{fig:summary} summarizes the relative WAN performance of the \sysname-based systems (star markers), compared with HotStuff (circle marker), when no faults occur, for different numbers of validators and workers (cross marker, number of workers after `W'). Throughput (x-axis) is increased with a single \sysname worker and vastly increased when leveraging the parallelizable nature of \sysname to a throughput of over 500,000 tx/sec, for a latency (y-axis) lower than 3.5 seconds. 


\section{Overview} \label{sec:overview}

This paper presents the design and implementation of \sysname, a
DAG-based \emph{\mempool abstraction}. 
%
\sysname ensures efficient wide
availability and integrity of user-submitted transactions under a fully
asynchronous network. 
It is a structured, persistent, Byzantine
fault-tolerant distributed storage that provides availability as well
as a partial order on blocks of transactions. 
In this section we define
the problem \sysname addresses, its system and security model, as well as a high-level overview and the main engineering challenges.

\subsection{System model, goals and assumptions}

We assume a message-passing system with a set of $n$ parties and 
a computationally bounded adversary that controls the network and can
corrupt up to $f < n/3$ parties.
We say that parties corrupted by the adversary are \emph{Byzantine} or \emph{faulty} and
the rest are \emph{honest} or \emph{correct}.
To capture real-world networks we assume asynchronous \emph{eventually
reliable} communication links among honest parties.
That is, there is no bound on message delays and there is a finite but unknown number of messages that can be lost.

Informally the \sysname \mempool exposes to all participants, a \emph{key-value} block store abstraction that can be used to read and write blocks of transactions and extract partial orders on these blocks. Nodes maintaining the \mempool are able to use the short key to refer to values stored in the shared store and convince others that these values will be available upon request by anyone.
The \sysname \mempool uses a round-based DAG structure that we describe in detail in the next sections. We first provide a formal definition of the \sysname \mempool. 

A \emph{block} $b$ contains a list of transactions and a list of
references to previous blocks. The unique 
(cryptographic) digest of its
contents, $d$, is used as its identifier to reference the block.
Including in a block, a reference encodes a causal
`happened-before' relation between
the blocks (which we denote $b \rightarrow b'$). The ordering of
transactions and references within the block also explicitly encodes
their order, and by convention, we consider all referenced blocks
happened before all transactions in the block.

Our \mempool abstraction supports a number of operations: 
A \emph{write(d,b)} operation stores a block $b$ associated with its
digest (key) $d$.
The returned value \emph{c(d)} represents an unforgeable
\emph{certificate of availability} on the digest $d$ and we say that the write \emph{succeeds} when \emph{c(d)} is formed.
A \emph{valid(d, c(d))} operation returns true if the certificate is valid, and false if it is not.
A \emph{read(d)} operation returns a block $b$ if a write(d,b) has succeeded.
A \emph{read\_causal(d)} returns a set of blocks $B$ such that $\forall
b' \in B \quad b' \rightarrow \ldots \rightarrow \text{read}(d)$, i.e.,
for every $b' \in B$, there is a transitive happened before relationship with $b$.

The \sysname \mempool satisfies the following properties:

\begin{itemize}[noitemsep,topsep=0pt]
     \item \textbf{Integrity:} For any certified digest $d$ every two
     invocations of \emph{read(d)} by honest parties that return a value, return the same value.

    \item \textbf{Block-Availability:} If a read operation \emph{read(d)} is invoked by an honest party after \emph{write}$(d,b)$ succeeds for an honest party, the \emph{read(d)} eventually completes and returns $b$.
    
    
    
     \item    \textbf{Containment:} Let $B$ be the set returned by a
    \emph{read\_cau\-sal(d)} operation, then for every $ b' \in B$, the
    set $B'$ returned by \emph{read\_causal(d')}, $B' \subseteq B$.
    
\item    \textbf{$2/3$-Causality:} A successful  \emph{read\_causal(d)} returns a set $B$ that contains at least $2/3$ of the blocks written successfully before \emph{write}$(d,b)$ was invoked.
   
       \item \textbf{$1/2$-Chain Quality} At least $1/2$ of the blocks in the returned set $B$ of a successful \emph{read\_causal(d)} invocation were  written by honest parties.

\end{itemize}
The Integrity and Block-Availability properties of \sysname allow us to clearly separate data dissemination from consensus.
That is, with \sysname, the consensus layer only needs to order block digest certificates, which which have a small size. 
Moreover, the Causality and Containment properties guarantee that any consensus protocol that leverages \sysname (even partially synchronous) achieves high-throughput despite periods of asynchrony.
This is because once we agree on a block digest, eg.\ when synchrony is restored, we can safely totally order all its causally ordered blocks created during the periods of asynchrony.
Therefore, \emph{with \sysname, different Byzantine consensus protocols mostly differ in the latency they achieve under different network conditions}, as we examine in Section~\ref{sec:evaluation}.
The Chain-Quality~\cite{DBLP:conf/eurocrypt/GarayKL15} property allows \sysname to be used by Blockchains providing censorship resistance.

Last but not least, our mempool abstraction can scale out and support the increasing demand in consensus services. 
Therefore, we aim to achieve the above theoretical properties and at the time satisfy the following:

\begin{itemize}[noitemsep,topsep=0pt]
    
    \item \textbf{Scale out:} \sysname's throughput increases linearly with the number of resources each validator has while the latency does not suffer.
    
\end{itemize}

\subsection{Intuitions behind the \sysname design }\label{sec:roadmap}

Established blockchains~\cite{nakamoto2008bitcoin,baudet2019state}
implement a best-effort gossip \mempool. A transaction submitted to one validator
is gossiped to all others. This leads to fine-grained double transmissions: most transactions are
shared first by the \mempool, and then the
miner/leader creates a block that re-shares them. In this section we extend step-by-step this basic design towards \sysname 
to (i) reduce the need for double transmission when leaders propose blocks, and (ii) enable scaling out when more resources are available.

A first step 
is to broadcast blocks instead of transactions and let the leader propose a hash of a block, relying on the \mempool layer to provide its \textbf{integrity-protected} content. 
However, validators also need to ensure hashes represent available blocks, requiring them to download them before certifying a block -- within the critical path of the consensus algorithm.

To ensure \textbf{availability}, as a second step, we consistently broadcast~\cite{cachin2011introduction} the block, 
resulting in a certificate that the block will be available for download. A leader proposes a certificate, 
which is short and proves the block will be available. 
However, one
certificate per \mempool block has to be included, and if the consensus temporarily loses
liveness then the number of certificates to be committed may grow indefinitely.

As a third step, we add \textbf{causality} to propose
a \emph{single certificate for multiple \mempool blocks}: 
\mempool blocks include certificates of past \mempool blocks, 
from all validators. As a result, a certificate refers, to a
block, and its full causal history. 
A leader proposing such a fixed-size
certificate, therefore, proposes an extension to the sequence containing
blocks from its full history. This design is extremely
economical of the leader's bandwidth, and ensures that delays in
reaching consensus impact latency but not average throughput--as
mempool blocks continue to be produced and are eventually
committed. Nevertheless, two issues remain: \first A very fast
validator may force others to perform large downloads by
generating blocks at a high speed; \second 
honest validators may not get enough bandwidth to share their blocks
with others -- leading to potential censorship.

A fourth step provides \textbf{Chain Quality} by imposing restrictions on block creation rate. 
Each block from a validator contains a round number, and must include a quorum of certificates from the previous round to be valid.
 As a result, a fraction of honest validators' blocks are included in any proposal. Additionally, a validator cannot advance to a \mempool round
 before some honest ones concluded the previous
 round, preventing flooding. As a result \sysname provides the consensus layer censroship-resistence (as defined in HoneyBadger BFT~\cite{miller2016honey}) without the need for using any additional mechanisms such as threshold encryption. Hence, our Narwhal-Hotstuff system is the only partially synchronous quorum based protocol that provides censorship-resistance. The adversary cannot kill leaders because it does not like the proposal as all proposals include at lest 50\% of honest transactions, even the ones from Byzantine leaders.
 
 The final fifth design step is that of enabling \textbf{scale-out}. 
 Instead of having a single machine creating \mempool blocks, multiple worker machines per validator can share \mempool sub-blocks, called \emph{batches}. One primary integrates references to them in \mempool primary blocks. This enables validators to commit a mass of computational, 
 storage, and networking resources to the task of sharing transactions--allowing for quasi-linear scaling.

\sysname is the culmination of the above five design steps, evolving the basic \mempool design to a robust and performant data dissemination and availability layer. \sysname can be used to off-load the critical path of traditional consensus protocols such as HotStuff or leverage the \textbf{containment property} to perform fully asynchronous consensus. We next study \sysname in more detail.

\section{\sysname Core Design} \label{sec:design}

In Section~\ref{sec:dag} we present the core protocol for a mempool and then in Section~\ref{sec:consensus} we show how to use it to get high-throughput consensus. Finally, in ~\Cref{sec:gc} we address the main roadblock of previous DAG designs, garbage collection.

\subsection{The \sysname \mempool}\label{sec:dag}

 \begin{figure}[t]
     \centering
     \includegraphics[width=0.85\columnwidth]{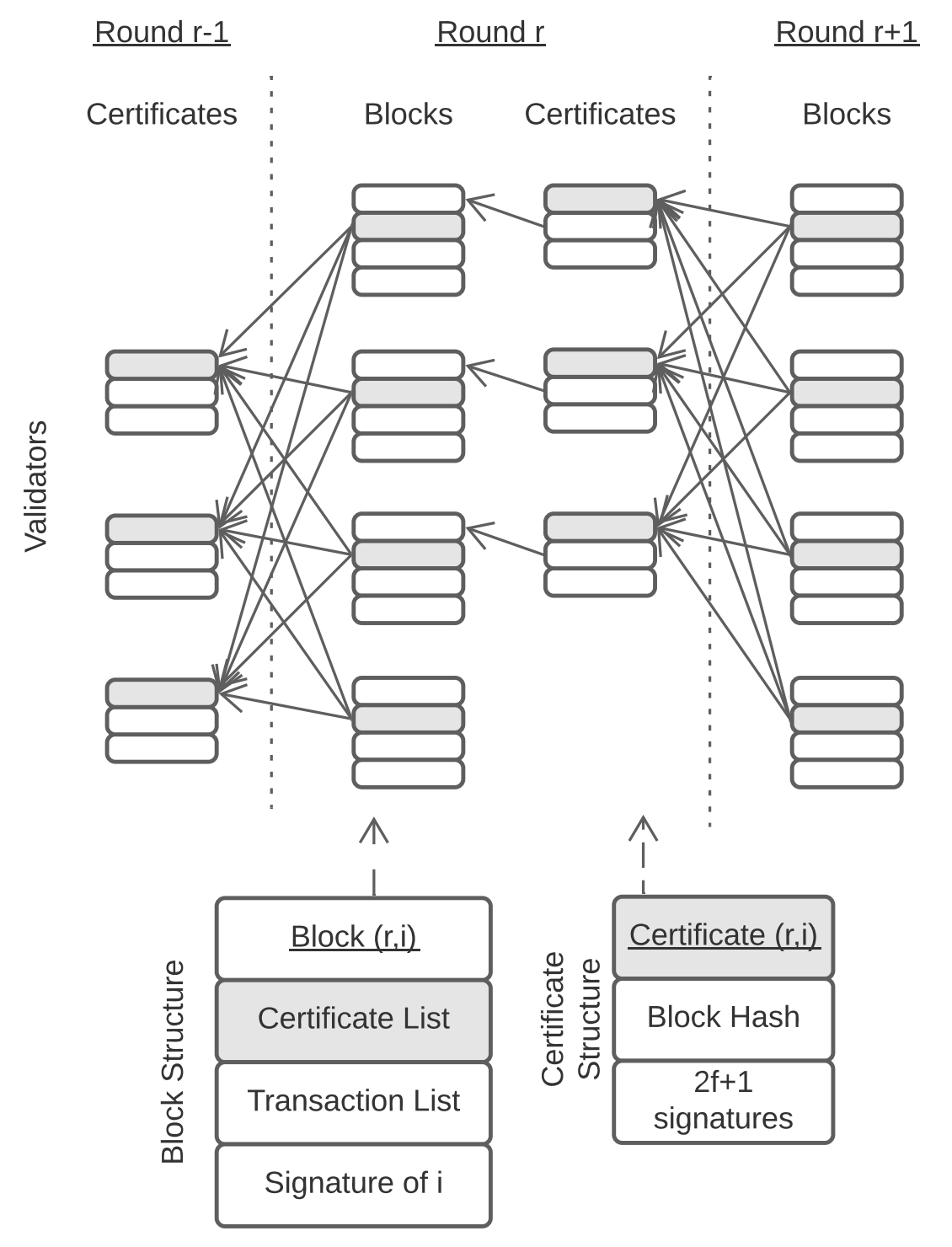}
     \caption{ Three rounds of \sysname. In round $r-1$ there are enough ($N-f$) certified blocks, and validators start building blocks for round $r$. Each includes a batch of transactions and $N-f$ certificates for round $r-1$. Blocks for $r$ can only include certificates from the previous round. 
     Once a validator has a ready block it broadcasts it to others to form a certificate and then shares the certificate with all other validators to include at round $r+1$.}
     \label{fig:dag}
\end{figure}

The \sysname \mempool is based on ideas from reliable broadcast~\cite{bracha1985asynchronous} and reliable storage~\cite{DBLP:conf/sosp/Abd-El-MalekGGRW05}. It additionally uses a Byzantine fault-tolerant version of Threshold clocks~\cite{tlc} as a pacemaker to advance rounds. An illustration of \sysname operation, forming a block DAG, can be seen in Figure~\ref{fig:dag}.





  
Each validator maintains the current local round $r$, starting at zero. Validators continuously receive transactions from clients and accumulate them into a transaction list (see Fig.~\ref{fig:dag}). They also receive certificates of
availability for blocks at $r$ and accumulate them into a certificate list.

Once certificates for round $r-1$ are accumulated from $2f+1$ distinct validators, a validator moves the local round to $r$, creates, and broadcasts a block for the new round. Each block includes the identity of its creator, and local round $r$, the current list of transactions and certificates from $r-1$, and a signature from its creator. Correct validators only create a single block per round.

The validators \emph{reliably broadcast}~\cite{bracha1985asynchronous} each block they create to ensure integrity and availability of the block.
For practical reasons we do not implement the standard push strategy that requires quadratic communication, but instead use a pull strategy to make sure we do not pay the communication penalty in the common case (we give more details in Section~\ref{sec:qb-rbc}).
In a nutshell, the block creator sends the block to all
validators, who check if it is \emph{valid} and then reply with their
signatures. 
A valid block must
\begin{enumerate}[noitemsep,topsep=1pt]
    \item contain a valid signature from its creator,
    \item be at the local round $r$ of the validator checking it,
    \item be at round $0$ (genesis), or contain certificates for at least $2f+1$ blocks of round $r-1$,
    \item be the first one received from the creator for round $r$.
\end{enumerate}
If a block is valid the other validators store it and acknowledge it by signing its block digest, round number, and creator's identity. We note that condition (2) may lead to blocks with an older logical time being dismissed by some validators. However, blocks with a future round contain $2f+1$ certificates that ensure a validator advances its round into the future and signs the newer block.
Once the creator gets $2f+1$ distinct acknowledgments for a
block, it combines them into a \emph{certificate of block availability
}, that includes the block digest, current round, and creator
identity.
Then, the creator sends the certificate to all other validators so that they can include it in their next block. 

The system is initialized through all validators creating and certifying empty blocks for round $r=0$. These blocks do not contain any transactions and are valid without reference to certificates for past blocks.

\vspace{3mm}
\noindent \textbf{Intuitions behind security argument.}
A certificate of availability includes $2f+1$ signatures, ie.\ at least $f+1$ honest validators have checked and stored the block.
Thus, the block is available for retrieval when needed to sequence transactions. 
Further, since honest validators have checked the conditions before signing the certificate, quorum intersection ensures Block-Availability and Integrity (i.e., prevent equivocation) of each block. Since a block contains references to certificates
previous rounds, we get by an inductive argument that all blocks in the
causal history are certified and available, satisfying causality\footnote{More complete security proofs for all properties of a mempool are provided in Appendix~\ref{sec:proofs}.}.

\subsection{Using \sysname for consensus}\label{sec:hs-over-dag}
\label{sec:consensus}

\Cref{fig:consensus} illustrates how \sysname can be combined with an eventually synchronous consensus protocol\footnote{The network assumption for this protocol to work is stronger than what \sysname requires.} (top) to enable high-throughput total ordering of transaction, which we elaborate on in \Cref{sec:hotstuff}; \sysname can also be augmented with a random coin mechanism to yield \consname (\Cref{fig:consensus}, bottom), a high-performance asynchronous consensus protocol we present in \Cref{sec:asynch}. Table~\ref{table:comp} summarizes the theoretical comparison of vanilla Hotstuff with \sysname-based systems, which we validate in our evaluation.

\newcolumntype{L}{>{\centering\arraybackslash}m{2.5cm}}
\begin{table}[t]
\small
\centering
\begin{tabular}{Lccc}
    \toprule 
    & HS  & \sysname-HS & \consname  \\
    \midrule
    
    Average-Case (Latency) & $3$ & $4$ & $4.5$ \\ 
    \addlinespace[0.5em]

    Worse-Case $f$  (Crashes Latency) & $O(n)$ & $O(n)$ & $4.5$ \\
    \addlinespace[0.5em]
    
    Asynchronous (Latency) & N/A & N/A & $7$  \\
    \addlinespace[0.5em]
    
    Unstable Network Throughput  & \xmark  & \cmark &  \cmark \\
    \addlinespace[0.5em]
    
    Asynchronous Throughput & \xmark & \xmark & \cmark \\
    \bottomrule
\end{tabular}
\caption{A comparison between Hotstuff and our protocols. We measure latency in RTTs (or certificates). Unstable Network is a network that allows for one commit between periods of asynchrony. By \cmark we mean the throughput is the same as it would be under synchrony}\label{table:comp}
\vspace{-0.8cm}
\end{table}

\subsubsection*{\sysname-Hotstuff.}\label{sec:hotstuff}

Consensus algorithms operating in partial synchrony, such as Hotstuff~\cite{DBLP:conf/podc/YinMRGA19} or LibraBFT~\cite{baudet2019state}, can leverage \sysname to improve their throughput. Such systems have a leader who proposes a block of transactions that is certified by other validators. Instead of proposing a block of transactions, a leader can propose one or more certificates of availability created in \sysname. Upon commit, the full uncommitted causal history of the certificates is deterministically ordered and committed. 
\sysname guarantees that given a certificate all validators see the same causal history, which is itself a DAG over blocks. 
As a result, any deterministic rule over this DAG leads to the same total ordering of blocks for all validators, achieving consensus.
Additionally, thanks to the availability properties of \sysname all committed blocks can be retrieved and transactions sequenced.

There are many advantages to leaders using \sysname over sending a block of transactions directly. Even in the absence of failures, a leader broadcasting transactions leads to uneven use of resources: the round leader has to use an enormous amount of bandwidth, while bandwidth at every other validator is underused.
In contrast, \sysname ensures bulk transaction information is efficiently and evenly shared at all times, leading to better network utilization and throughput.

Eventually-synchronous consensus protocols cannot provide liveness
during asynchrony periods or when leaders are Byzantine.
Therefore, with a naive mempool implementation, overall consensus throughput
goes to zero during such periods.
\sysname, in contrast, continues to share blocks and form certificates of availability even under asynchronous
networks, so blocks are always certified at maximal throughput.
Once the
consensus protocol manages to commit a digest, validators also commit
its causal history, with no gaps for periods of asynchrony. Nevertheless, an eventually synchronous protocol
still forfeits liveness during periods of asynchrony, leading to
increased latency. We show how to overcome this problem with \consname.

 \begin{figure}[t]
     \centering
     \includegraphics[width=\columnwidth]{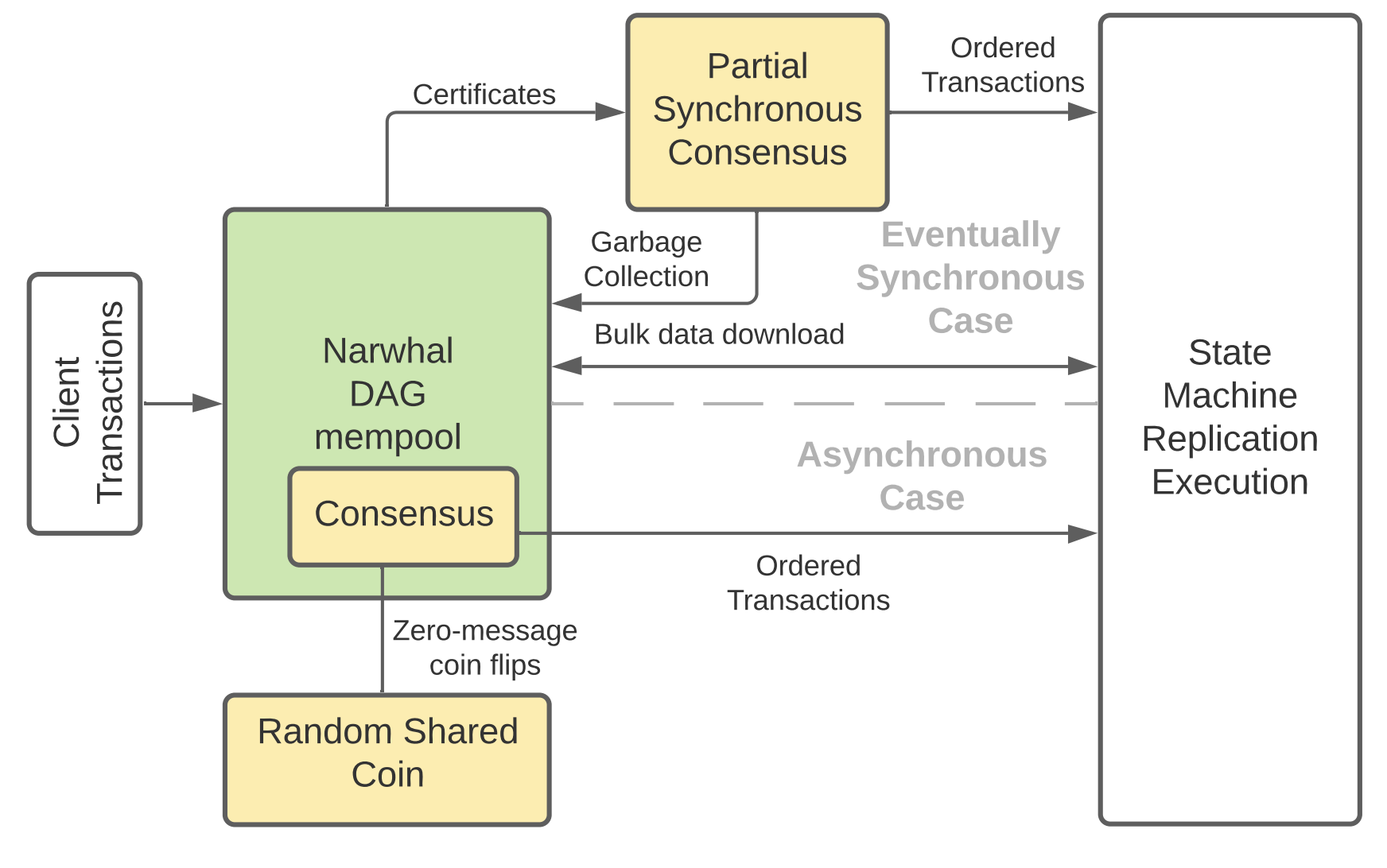}
     \caption{ Any consensus protocol can execute over the mempool by occasionally ordering certificates to \sysname blocks. \sysname guarantees their availability for the SMR execution.
     Alternatively, \sysname structure can be interpreted as an asynchronous consensus protocol with the (zero-message cost) addition of a random-coin.}
     \label{fig:consensus}
\end{figure}

\subsection{Garbage Collection}\label{sec:gc}

Another theoretical contribution of \sysname paired with any consensus algorithm is lifting one of the
main roadblocks to the adoption of DAG-based consensus algorithms (e.g.,
Hashgraph~\cite{baird2016swirlds}), garbage collection\footnote{
A bug in our garbage collection led to exhausting 120GB of RAM in minutes compared to 700MB memory footprint of \sysname.}. This challenge stems from
the fact that a DAG is a local structure and although it will
eventually converge to the same version in all validators there is no
guarantee on when this will happen. As a result, validators may
have to keep all blocks and certificates readily accessible
to (1) help their peers catch up and (2) be able to process arbitrary
old messages.

This is not a problem in \sysname, where we impose a strict round-based
structure on messages.
This restriction allows validators to decide on the validity of a block
only from information about the current round (to ensure uniqueness of
signed blocks).
Any other message, such as certified blocks, carries enough information
for validity to be established only with reference to cryptographic
verification keys.
As a result, validators in \sysname are not required to examine the
entire history to verify new blocks.
However, note that if two validators garbage collect different rounds
then when a new block $b$ is committed, validators might disagree on
$b$'s causal history and thus totally order different
histories.
To this end, \sysname leverages the properties of a consensus protocol
(such as the one we discuss in the previous section) to agree on the
garbage collection round.
Blocks from earlier rounds can be safely be stored off the main validator
and all later messages from previous rounds can be safely ignored. 

At a first sight this might look like an adversary can censor transactions by delaying them enough for the garbage collection to kick in and remove them from being actively sent. This is not true, an honest node that garbage collect an old round that didn't make it into the DAG re-inject transactions to a later round. Hence although the actual DAG blocks might be censored, eventually all transactions will be included in blocks and due to the 1/2-Chain Quality property they will be added to the consensus output within constant rounds.

All in all, validators in \sysname can operate with a fixed size
memory.
That is, $O(n)$ in-memory usage on a validator, containing blocks and
certificates for the current round, is enough to operate correctly.
Since certificates ensure block availability and integrity,
storing and servicing requests for blocks from previous rounds can be
offloaded to a passive and scalable distributed store or an external
provider operating a Content Distribution Network (CDN) such as
Cloudflare or S3. Protocols using the DAG content as a mempool for
consensus can directly access data from the CDN after sequencing to
enable execution of transactions using techniques from deterministic 
databases~\cite{abadi2018overview}.




\section{Building a Practical System}

In this section, we discuss two key practical challenges we had to address in order to enable \sysname to reach its full theoretical potential.

\subsection{Quorum-based reliable broadcast}\label{sec:qb-rbc}

In real-world reliable channels, like TCP, all state is lost and re-transmission ends if a connection drops. Theoretical reliable broadcast protocols, such as double-echo~\cite{cachin2011introduction}, rely on perfect point-to-point channels that re-transmit the same message forever, or at least until an acknowledgment, requiring unbounded memory to store messages at the application level. Since some validators may be Byzantine, acknowledgments cannot mitigate the denial-of-service risk. 

To avoid the need for perfect point-to-point channels we take advantage of the fault tolerance and the replication provided by the quorums we rely on to construct the DAG.
In the \sysname implementation each validator broadcasts a block for each round $r$: Subject to conditions specified, if $2f+1$ validators receive a block, they acknowledge it with a signature. $2f+1$ such signatures form a certificate of availability, that is then shared, and potentially included in blocks at round $r+1$. Once a validator advances to round $r+1$ it \emph{stops re-transmission and drops all pending undelivered messages} for rounds smaller than $r+1$.

A certificate-of-availability does not guarantee the totality property\footnote{The properties of reliable broadcast are Validity, No duplication, Integrity, Consistency, and Totality (see page 112 and 117 of~\cite{cachin2011introduction}).} needed for reliable broadcast: it may be that some honest nodes receive a block but others do not. However, if a block at round $r+1$ has a certificate-of-availability, the totality property can be ensured for all $2f+1$ blocks with certificates it contains for round $r$. Upon, receiving a certificate for a block at round $r+1$ validators can request all blocks in its causal history from validators that signed the certificates: since at least $f+1$ honest validators store each block, the probability of receiving a correct response grows exponentially after asking a handful of validators. This \emph{pull} mechanism is DoS resistant and efficient: At any time only $O(1)$ requests for each block are active, and all pending requests can be dropped after receiving a correct response with the sought block. This happens within $O(1)$ requests on average, unless the adversary actively attacks the network links, requiring $O(n)$ requests at most, which matches the worst-case theoretical lower bound~\cite{dolev1985bounds}. 

The combination of block availability certifications, their inclusion in subsequent blocks, 
and a `pull mechanism' to request missing certified blocks leads to a reliable broadcast protocol. 
Storage for re-transmissions is bounded by the time it takes to advance a round and the time it takes to retrieve a 
certified block -- taking space bounded by $O(n)$ in the size of the quorum (with small constants).

\subsection{Scale-Out Validators}

 \begin{figure}[t]
     \centering
     \includegraphics[width=0.7\linewidth]{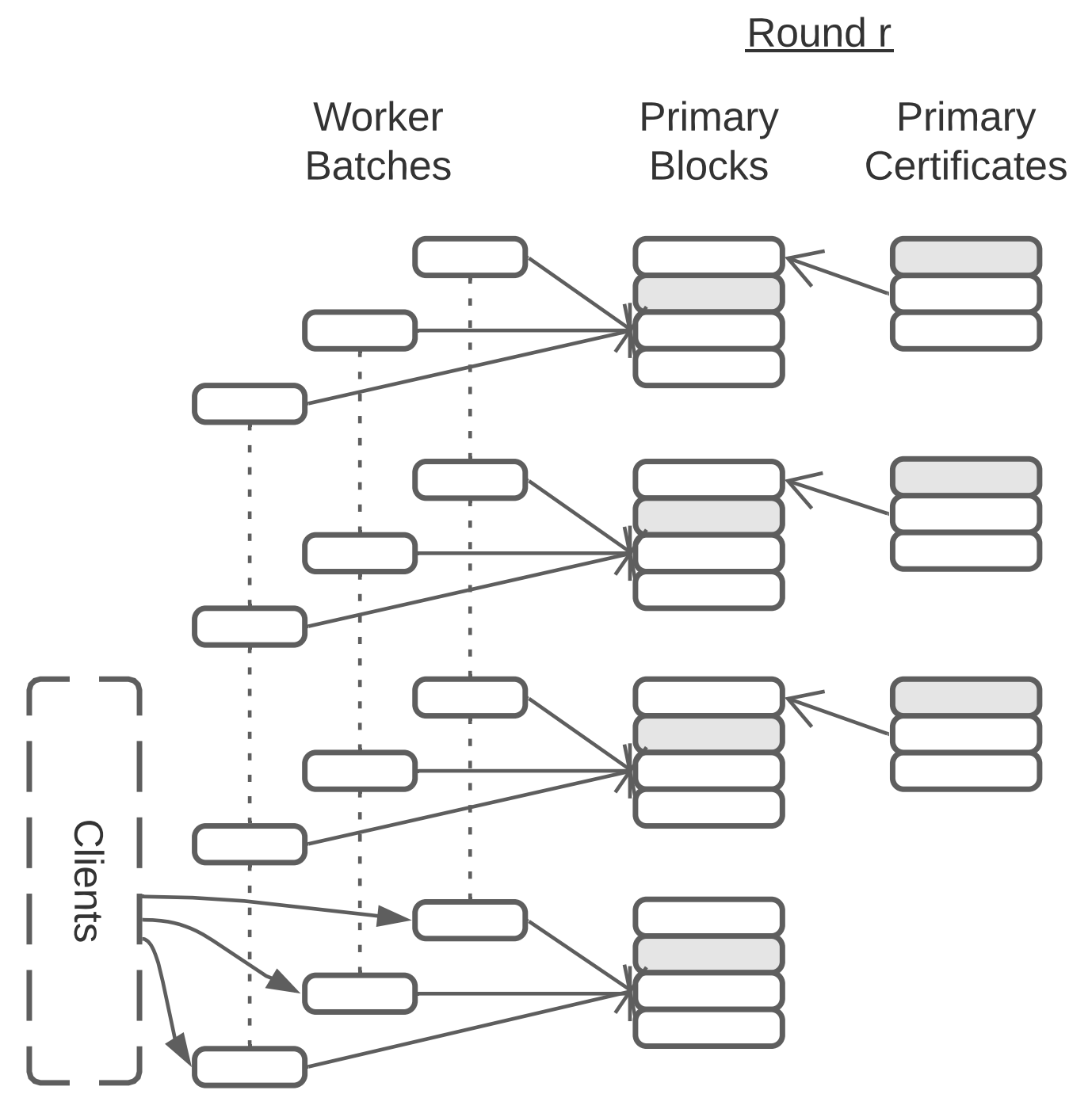}
          \vspace{-0.3cm}
     \caption{ Scale-Out Architecture in \sysname. Within each validator, there is one primary machine that handles the meta-data of building the DAG and multiple worker machines (3 shown in the example) each one streaming Transactions batches to other validators. The hashes of the batches are sent to the primary that includes them in the next block. Clients send transactions to worker machines at all validators. }
     \label{fig:sharding}
     \vspace{-0.5cm}
\end{figure}

In a consensus system all correct validators eventually need to receive all sequenced transactions. \sysname, like any other mempool, is therefore ultimately limited by the bandwidth, storage, and processing capabilities of a single validator.  However, a validator is a unit of authority and trust, and does not have to be limited to employing the resources of a single computer. We, therefore, adapt \sysname to use many computers per validator, to achieve a scale-out architecture.


\vspace{2mm}
\noindent \textbf{Core Scale-Out Design.}
We adapt \sysname to follow a simple primary-worker architecture as seen in Figure~\ref{fig:sharding}. We split the protocol messages into transaction data and metadata. Transferring and storing transaction data is an `embarrassingly parallel' process: A load balancer ensures transactions data are received by all workers at a similar rate; a worker then creates a batch of transactions, and sends it to the worker node of each of the other validators; once an acknowledgment has been received by a quorum of these, the cryptographic hash of the batch is shared with the primary of the validator for inclusion in a block.

The primary runs the \sysname protocol as specified, but instead of including transactions into a block, it includes cryptographic hashes of its own worker batches. The validation conditions for the reliable broadcast at other validators are also adapted to ensure availability: a primary only signs a block if the batches included have been stored by its own workers. This ensures, by induction, that all data referred to by a certificate of availability can be retrieved. 

A pull mechanism has to be implemented by the primary to seek missing batches: upon receiving a block that contains such a batch, the primary instructs its worker to pull the batch directly from the associated worker of the creator of the block. This requires minimal bandwidth at the primary. The pull command only needs to be re-transmitted during the duration of the round of the block that triggered it, ensuring only bounded memory is required.

\vspace{2mm}
\noindent \textbf{Streaming.} Primary blocks are much smaller when hashes of batches instead of transactions are included. Workers constantly create and share batches in the background. Small batches, in the order of a few hundred to a few thousand transactions (approximately 500KB), ensure transactions do not suffer more than some maximum latency. As a result, most of the batches are available to other validators before primary blocks arrive. This reduces latency since (1) there is less wait time from receiving a primary block to signing it and (2) while waiting to advance the round (since we still need $2f+1$ primary blocks) workers continue to stream new batches to be included in the next round's block.

\vspace{2mm}
\noindent \textbf{Future Bottlenecks.} At high transaction rates, a validator may increase capacity through adding more workers, or increasing batch sizes. Yet, eventually, the size of the primary blocks will become the bottleneck, requiring a more efficient accumulator such as a Merkle Tree root batch hashes. 
This is a largely theoretical bottleneck, and in our evaluation we never managed to observe the primary being a bottleneck. As an illustration: a sample batch size of 1,000 transactions of  512B each, is 512KB. The batch hash within the primary block is a minuscule 32B and 8B for meta-data (worker ID), i.e.\ 40B. This is a volume reduction ratio of 1:12, and we would need about 12,000 workers before the primary handles data volumes similar to a worker. So we leave the details of more scalable architectures as distant future work.

\section{\consname asynchronous consensus}\label{sec:asynch}

In this section, we present \consname, an asynchronous consensus algorithm for \sysname that remains live under 
asynchrony or DDoS attacks. \consname's theoretical starting point is DAG-Rider~\cite{DBLP:journals/corr/abs-2102-08325}, from which it inherent its safety guarantees. \consname modifies DAG-Rider into an implementable system and improves its latency in the common case. 

\consname validators operate a \sysname mempool as described in \Cref{sec:dag}, but also include in each of their blocks information to generate a 
distributed perfect random coin. Such a coin can be constructed from an adaptively secure
threshold signature scheme~\cite{DBLP:journals/joc/BonehLS04} for which a key setup can also be performed under full asynchrony~\cite{DBLP:conf/ccs/Kokoris-KogiasM20}.
The \consname algorithm gets the causally ordered DAG constructed by \sysname and totally orders its blocks with zero extra communication.
That is, every validator locally interprets its local view of the DAG and use the shared randomness to determine the total order.

To interpret the DAG, validators divide it into \emph{waves}, each of which consists of 3 consecutive rounds.
Conceptually, in the first round of a wave each validator proposes its block (and consequently all its causal history); in the second round each validator votes on the proposals by including them in their block; and in the third round validators produce randomness to elect one random leader's block in retrospect. 
Once the random coin for the wave is revealed, each validator $v$ commits the elected leader block $b$ of a wave $w$ if there are $f+1$ blocks in the second round of $w$ (in $v$'s local view of the DAG) that refer to $b$.
In which case, $v$ then carefully orders $b$'s causal history up to the garbage collection point.
However, since validators may obtain different local views of the DAG, some may commit a leader block in a wave while others may not.

To guarantee that all honest validators eventually order the same block leaders, after committing the leader in a wave $w$, $v$ sets it to be the next \emph{candidate} to be ordered and recursively go back to the last wave $w'$ in which it committed a leader. 
For every wave $i$ in between $w'$ and $w$, $v$ checks whether there is a path between the candidate leader and the leader block $b$ of wave $i$. In case there is such a path $v$ orders $b$  before the current candidate, set the correct candidate to be $b$ and continue the recursion.
This mechanism's safety argument is similar to DAG-Rider~\cite{DBLP:journals/corr/abs-2102-08325}, and we provide here only an intuition.

Our theoretical main difference from DAG-Rider is the termination guarantee.
While both protocols commit a block leader every constant number of rounds in expectation, \consname's constants are better in the common case.
By common case we mean networks where message delays are randomly distributed rather then controlled by the adversary.
Specifically, DAG-Rider's waves consist of 4 rounds, and thus each block in the DAG is committed in expectation every $5.5$ rounds in the common case.
In \consname we improve latency by considering waves that consists of 3 rounds.
In addition, as an optimization, we piggyback the the first round of wave with the third round (that produces randomness) of the preceding one.
As a result, each block in the DAG is committed in expectation every $4.5$ rounds in the common case.

\Cref{fig:coin} illustrates two 
waves of \consname -- the first wave is rounds 1-3, and the second is rounds 3-5 -- with 4 validators and $f=1$.
The elected leaders of waves 1 and 2 are determined at rounds 3 and 5, and we denote them $L_1$ and $L_2$, respectively. 
There are an insufficient number of blocks in round 2 (less than $f+1$) that refer to ("vote for") $L_1$ and thus $L_1$ is not committed when round 3 is interpreted. 
However $f+1 =2$ blocks in round 4 refer to $L_2$, and as a result $L_2$ is eventually committed. 
Since there is a path between $L_2$ and $L_1$, $L_1$ is ordered before $L_2$. 
Meaning that the sub-DAG
causally dependent on $L_1$ is ordered first (by some deterministic rule), and then the same rule is applied to the sub-DAG causally dependent on $L_2$.

It is noteworthy that neither the random
shared coin generation nor the consensus logic introduces any additional messages over \sysname. 
Therefore, \consname  has zero message
overhead, and the same theoretical throughput as \sysname.

\subsection{Safety intuition}

Each round in the DAG contains of at least $2f+1$ blocks, and since
any quorum of $2f+1$ blocks intersect with any quorum of
$f+1$ blocks we get:
\begin{restatable}{lem}{quorum}
\label{lem:quorum}
If an honest validator commits a leader block $b$ in an wave $i$,
then any leader block $b'$ committed by any honest validator $v$ in a
future wave have a path to $b$ in $v$'s local DAG.
\end{restatable}

Given the above lemma, the recursive mechanism that is described above to order block leaders guarantees the following: 

\begin{restatable}{lem}{safety}
\label{lem:safty}

Any two honest validators commit the same sequence of block leaders.

\end{restatable}

Since after ordering a block leader each validator orders the leader's causal history by some pre-define deterministic rule, we get that by the Containment property of Narwhal all honest validators agree on the total order of the DAG's blocks.

\subsection{Liveness and latency}

\com{

In every wave they randomly elect a block leader from the first round of the wave and check if the commit rule is guaranteed in the forth round.
In which case, they carefully commit the leader's causal history.
To guarantee the  validity (fairness) of the atomic broadcast, i.e., ensuring all blocks are eventually added to the DAG, DAG-Rider uses weak links.
Each weak link, included in a block for a round $r$, refers to a hash representing a block for a round older than $r-1$ that otherwise is no included in the DAG.
For space limitation, for detailed description of DAG-Rider we refer to the paper~\cite{DBLP:journals/corr/abs-2102-08325}.

For the core algorithmic part we extend DAG-Rider: 1)~we replace the classic reliable broadcast with our version described in \Cref{sec:qb-rbc}; 2)~remove weak links to allow garbage collection; and 3)~change the commit rule to have a better common-case latency. All these changes result in \consname.

\subsubsection*{Background:}
DAG-Rider~\cite{DBLP:journals/corr/abs-2102-08325} is a state-of-the-art DAG-based atomic broadcast protocol with expected constant latency and asymptotically optimal amortized  communication complexity.
They use a black-box reliable broadcast abstraction to construct a DAG structure similar to the high-level of Narwhal, which validators locally interpret to totally order all DAG blocks.

To interpret the DAG, validators divide it into \emph{waves}, each of which consists of 4 rounds.
In every wave they randomly elect a block leader from the first round of the wave and check if the commit rule is guaranteed in the forth round.
In which case, they carefully commit the leader's causal history.
To guarantee the  validity (fairness) of the atomic broadcast, i.e., ensuring all blocks are eventually added to the DAG, DAG-Rider uses weak links.
Each weak link, included in a block for a round $r$, refers to a hash representing a block for a round older than $r-1$ that otherwise is no included in the DAG.
For space limitation, for detailed description of DAG-Rider we refer to the paper~\cite{DBLP:journals/corr/abs-2102-08325}. 

\subsubsection*{Weak links and GC:} 
%

Weak links make garbage collection impossible, as every block ever received has to be stored in-case it is pointed by a weak link. Thus, the strong notion of fairness of DAG-Rider is unimplementable within finite storage. In \consname we forbid weak links. We instead re-inject transactions of uncommitted garbage collected blocks into new blocks in subsequent rounds. Thus, Instead of block-level fairness this achieves the more relevant metric of transaction-level fairness. 



In a nutshell, \sysname can be interpreted as a round-based DAG of
blocks of transactions.
Due to the reliable broadcast of each block, all honest
validators eventually interpret the same DAG (even in the presence of
Byzantine validators).
We use the idea from VABA~\cite{DBLP:conf/podc/AbrahamMS19} to let the asynchronous adversary
commit to a communication schedule and then leverage a perfect shared
coin to get a constant probability to agree on a useful work
produced therein.  

In \consname, validators interpret every three rounds of the \sysname DAG as a
\emph{consensus instance}.
The links in the blocks of the first two rounds are interpreted as
all-to-all message exchanges and the third round\footnote{To improve latency we combine the third round with
the first round of the next instance.} produces a shared perfect coin
that is used to elect a unique block from the first round to be the
leader of the instance. 
The goal is to interpret the DAG in a way that, with a constant
probability, safely commit the leader of each instance.
Once a leader is committed, its entire causal history in the DAG can
be safely totally ordered by any deterministic rule.

\subsubsection*{Commit rule:} The commit rule for \consname is simple and heavily inspired by DAG-Rider; to get Tusk we simply need to change Algorithm 3 of DAG-Rider to the code in \Cref{alg:SMROnDag} and reduce the wave length to $2$ round from $4$.
\begin{algorithm}[h!]
	\caption{\textbf{ Commit Rule for \consname}}
	\label{alg:SMROnDag}
	\small
	\begin{algorithmic}[1]

		\Upon{$\textit{wave\_ready}(w)$} 
		\label{alg:SMR:waveCompletion} 
		\State $v \gets \textit{get\_wave\_vertex\_leader}(w)$ \label{alg:SMR:getWaveLeader}
		\If{$v = \bot \vee \left|  \left\{ v' \in DAG_i[\textit{round}(w,2)] \colon
		\textit{strong\_path}(v',v) \right\} \right| < f+1$ } 
	
		\label{alg:SMR:commitrule}
		  \State \Return
		\EndIf

        \State $\textit{leadersStack}.\text{push}(v)$ \label{alg:SMR:stackPush1}
        \For{wave $w'$ from $w - 1$ down to
        $\textit{decidedWave} + 1$} \label{alg:SMR:orderLadersForLoop}
        \State $v' \gets \textit{get\_wave\_vertex\_leader}(w')$
        \If{$v' \neq \bot \wedge \textit{strong\_path}(v,v')$}
            \label{alg:DAGabstraction:checkpath}
            \State $\textit{leadersStack}.\text{push}(v')$ \label{alg:SMR:stackPush2}
            \State $v \gets v'$
            \label{alg:SMR:orderLadersForLoopEnd}
            
        \EndIf
        \EndFor
        \State $\textit{decidedWave} \gets w$ \label{alg:SMR:decidedWaveUpdate}
		  \State $\textit{order\_vertices}(\textit{leadersStack})$
		
		\EndUpon

	\end{algorithmic}
\end{algorithm}
A validator commits a block leader $b$ of a wave $i$ if its local
view of the DAG includes at least $f+1$ blocks in the second round of wave $i$ with links to $b$.
In order to defend against adaptive adversaries we only flip the coin
for wave $i$ at the first round of wave $i+1$, as a result the latency is $3$ rounds (instead of DAG-Rider's $4$).
\sasha{Wait, our amortized latency is better than 3 rounds, right?}
\lefteris{How? I think it is slightly worse in practice}
\Cref{fig:coin} illustrates two 
rounds of \consname (five rounds of \sysname): in round 3 an insufficient 
number of blocks (ie.\ one instead of two or more) support the leader at round 1 (L1); in round 5 sufficient 
blocks (ie.\ two) support the leader block at round 3 (L2) and as a result both 
L1 and L2 are sequenced. A deterministic algorithm on the sub-DAG
causally dependent on L1 and then L2 is then applied to sequence all blocks 
and transactions.

Validators may have different local DAGs at the time
they interpret wave $i$, and some may commit $b$
while others do not.
However, any quorum of $2f+1$ blocks intersect with any quorum of
$f+1$ blocks, and we prove in Appendix~\ref{sec:proofs} that:
\begin{restatable}{lem}{quorum}
\label{lem:quorum}
If an honest validator commits a leader block $b$ in an wave $i$,
then any leader block $b'$ committed by any honest validator $v$ in
future waves have a path to $b$ in $v$'s local DAG.
\end{restatable}
To guarantee that honest validators commit the same block leaders in the
same order, we adopt DAG-Rider's approach: when the
commit rule is satisfied for the leader of some wave, we recursively go back to the last
wave in which we committed a leader, and for every wave $i$ in
between we check whether there is a path between the leader of the
current wave and the leader of $i$. In case there is such a path we
commit the leader of wave $i$ before the leader of the current
wave.
We use Lemma~\ref{lem:quorum} to prove the following:

\begin{restatable}{lem}{safety}
\label{lem:safty}

Any two honest validators commit the same sequence of block leaders.

\end{restatable}

Once a validator commits a block leader it deterministically commits
all its causal history. 
By the Containment property of the DAG and the fact the eventually all
validators have the same DAG (due to reliable broadcast), it is
guaranteed that all honest validators agree on the total order.

The Safety proof of Tusk follows from the Safety proofs of DAG-Rider. 

}

Inspired from the VABA~\cite{DBLP:conf/podc/AbrahamMS19} protocol we let an adversary
commit to a communication schedule and then leverage a perfect shared
coin to get a constant probability, despite asynchrony, to agree on a useful work
produced therein.  

As for termination, we use a combinatorial argument to prove in
Appendix~\ref{sec:proofs}:

\begin{restatable}{lem}{livenessaux}
\label{lem:livenessaux}

For every wave $w$ there are at least $f+1$ blocks in the first
round of $w$ that satisfy the commit rule.

\end{restatable}

To guarantee liveness against an adaptive asynchronous adversary, 
we use the randomness produced in the third round of a wave to determine the wave's block leader.
Therefore, the adversary learns who is the block leader of a
wave only after the first two rounds of the wave are fixed. 
Thus the $f+1$ blocks that satisfy the commit rule (from the above lemma) are determined before the adversary learns the block leader.
Therefore, the probability to commit a block leader in each wave is at least $\frac{f+1}{3f+1} > 1/3$ even if the adversary fully controls the network.

In Appendix~\ref{sec:proofs} we prove the following lemma:  

\begin{restatable}{lem}{liveness}
\label{lem:liveness}

In expectation, \consname commits a block leader
every 7 rounds in the DAG under an asynchronous adversary.

\end{restatable}

As for latency in realistic networks in which message delays are distributed uniformly at random, we prove in Appendix~\ref{sec:proofs} the following:

\begin{restatable}{lem}{livenessran}
\label{lem:livenessran}

In networks with random message delays, in expectation, \consname commits each block in the DAG in 4.5 rounds.

\end{restatable}

 \begin{figure}[t]
     \centering
     \includegraphics[width=\columnwidth]{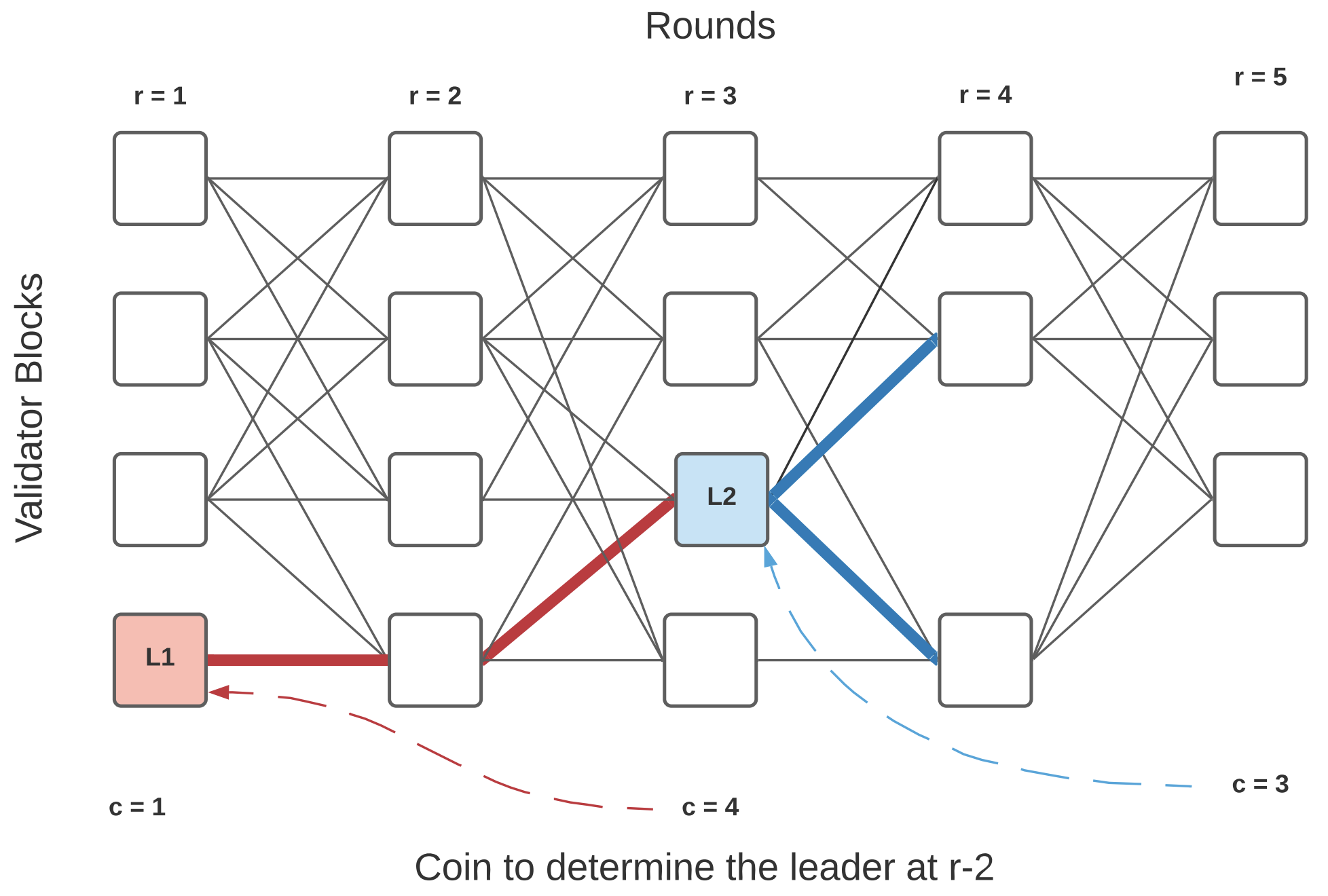}
     \caption{Example of commit rule in \consname. Every odd round has a coin value that selects a leader of round $r-2$. If the leader has less than $f+1$ support (red) they are ignored, otherwise (blue) the algorithm searches the causal DAG to commit all preceding leaders (including red) and totally orders the rest of the DAG afterward.
     }
     \label{fig:coin}
\end{figure}

%

\section{Implementation} \label{sec:implementation}
We implement a networked multi-core \sysname validator in Rust, using Tokio\footnote{\url{https://tokio.rs}} for asynchronous networking, ed25519-dalek\footnote{\url{https://github.com/dalek-cryptography/ed25519-dalek}} for elliptic curve based  signatures. Data-structures are persisted using RocksDB\footnote{\url{https://rocksdb.org}}. 
We use TCP to achieve reliable point-to-point channels, necessary to correctly implement the distributed system abstractions. We keep a list of messages to be sent between peers in memory and attempt to send them through persistent TCP channels to other peers. In case TCP channels are drooped we attempt to re-establish them and attempt again to send stored messages. Eventually, the primary or worker logic establishes that a message is no more needed to make progress, and it is removed from memory and not re-sent -- this ensures that the number of messages to unavailable peers does not become unbounded and a vector for Denial-of-Service. 
The implementation is around 4,000 LOC and a further 2,000 LOC of unit tests. 
We are open sourcing the Rust implementation of \sysname
and \consname\footnote{
\ifdefined\cameraReady
\url{https://github.com/facebookresearch/narwhal/tree/tusk}
\else
Link omitted for blind review.
\fi
}, HS-over-\sysname\footnote{
\ifdefined\cameraReady
\url{https://github.com/facebookresearch/narwhal/tree/narwhal-hs}
\else
Link omitted for blind review.
\fi
} as well as all Amazon Web Services orchestration scripts, benchmarking scripts, and measurements data to enable reproducible results\footnote{
\ifdefined\cameraReady
\url{https://github.com/facebookresearch/narwhal/tree/tusk/results/data}
\else
Link omitted for blind review.
\fi
}.

We evaluate both \consname (Section~\ref{sec:asynch}) and HS-over-\sysname  (Section~\ref{sec:hotstuff}).
Additionally, to have a fair comparison we implement Hotstuff, but unlike the original paper we \first add persistent storage in the nodes (since we are building a fault-tolerant system), \second evaluate it in a WAN, and \third implement the pacemaker module that is abstracted away following the LibraBFT specification~\cite{baudet2019state}.
We specify two versions of HotStuff (HS). First `baseline-HS' implements the standard way blockchains (Bitcoin or Libra) disseminate single transactions on the gossip/broadcast network. Second Batched-HS implements the state-of-the-art technique~\cite{graphene,prism,CastroL02} of validators batching transactions and sending them out of the critical path. 
Specifically, Batched-HS separates the task of data dissemination and consensus in the same way as Prism~\cite{prism}. It first disseminates batches of transactions (called `transaction blocks' in Prism), then the leader proposes hashes of batches to amortize the cost of the initial broadcast. These transaction batches perform a similar function to \sysname's headers.
The goal of this version is to show that this solution already gives benefits in a stable network but is not robust enough for a real deployment.
We also open source our implementation of Batched-HS\footnote{
\ifdefined\cameraReady
\url{https://github.com/asonnino/hotstuff/tree/d771d4868db301bcb5e3deaa915b5017220463f6}
\else
Link omitted for blind review.
\fi
}.
\section{Evaluation} \label{sec:evaluation}
We evaluate the throughput and latency of our implementation of \sysname through experiments on AWS. 
We particularly aim to demonstrate that \first \sysname as a \mempool has advantages over the existing simple \mempool as well as straightforward extensions of it and \second that the scale-out is effective, in that it increases throughput linearly as expected. Additionally, we want to show that 
\third \consname is a highly performing consensus protocol that leverages \sysname to maintain high throughput when increasing the number of validators (proving our claim that message complexity is not that important), as well as that \sysname, provides \fourth robustness when some parts of the system inevitably crash-fail or suffer attacks. For an accompanying theoretical analysis of these claims see \Cref{table:comp}. Note that evaluating BFT protocols in the presence of Byzantine faults is still an open research question~\cite{twins}.

We deploy a testbed on Amazon Web Services, using \texttt{m5.8xlarge} instances across 5 different AWS regions: N. Virginia (us-east-1), N. California (us-west-1), Sydney (ap-southeast-2), Stockholm (eu-north-1), and Tokyo (ap-northeast-1). They provide 10Gbps of bandwidth, 32 virtual CPUs (16 physical core) on a 2.5GHz, Intel Xeon Platinum 8175, and 128GB memory and run Linux Ubuntu server 20.04. 

In the following sections, each measurement in the graphs is the average of 2 runs, and the error bars represent one standard deviation. 
Our baseline experiment parameters are: 4 validators each running with a single worker, a batch size of 500KB, a block size of 1KB, a transaction size of 512B, and one benchmark client per worker submitting transactions at a fixed rate for a duration of 5 minutes. We then vary these baseline parameters through our experiments to illustrate their impact on performance. When referring to \emph{latency}, we mean the time elapsed from when the client submits the transaction to when the transaction is committed by the leader that proposed it as part of a block. We measure it by tracking sample transactions throughout the system, under high load.

\subsection{\sysname as a \mempool}

\begin{figure*}[t]
    \centering
    \includegraphics[width=\textwidth]{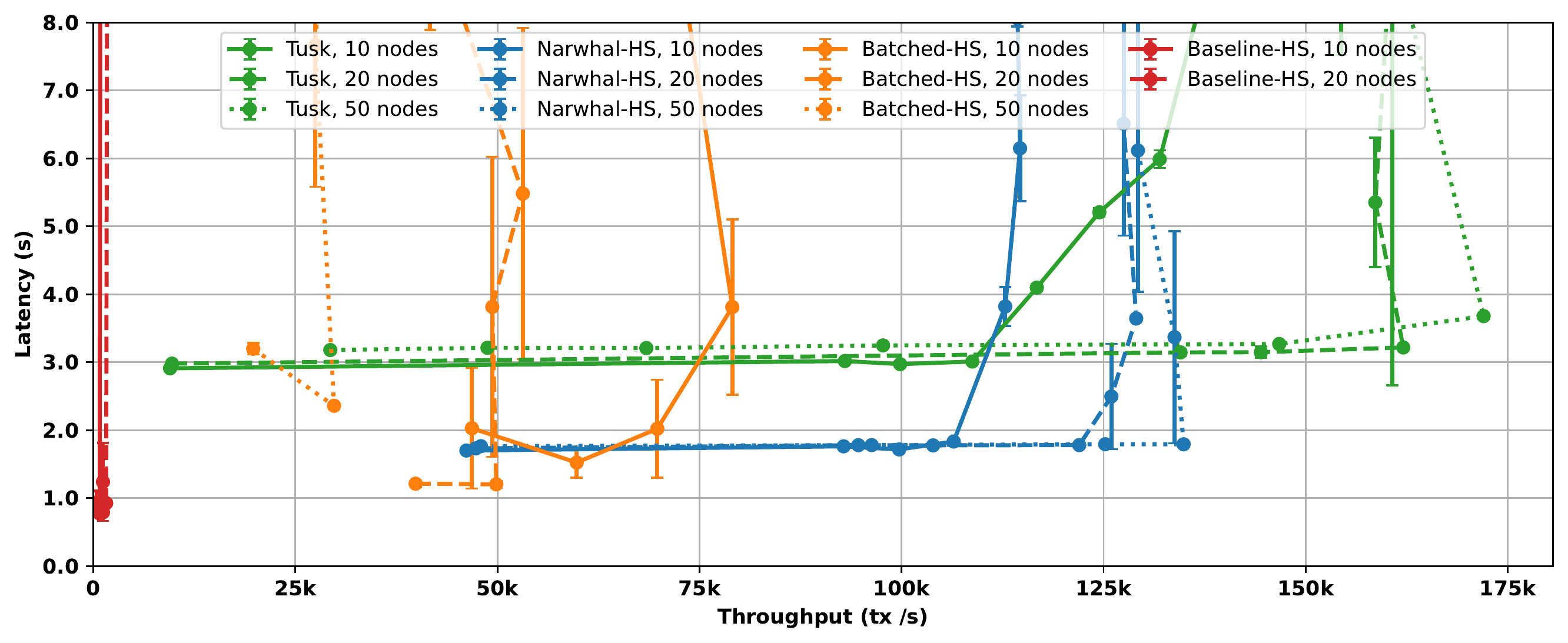}
    \caption{Comparative throughput-latency performance for the novel \sysname-HotStuff, \consname, batched-HotStuff and the baseline HotStuff. WAN measurements with 10, 20, and 50 validators, using 1 worker collocated with the primary. No validator faults, 500KB max.\ block size and 512B transaction size.}
    \label{fig:latency}
\end{figure*}

\Cref{fig:latency} illustrates the throughput of Narwhal for varying numbers of validators, as well as baseline and batched HS:

\para{Baseline HS}
Baseline HS throughput (see \Cref{fig:latency}, Baseline-HS lines, left bottom corner), with a naive mempool as originally proposed, is quite low. With either 10 and 20 validators throughput never exceeds 1,800 tx/s, although latency at such low throughput is very good at around 1 second. Such surprisingly low numbers are comparable to other works~\cite{DBLP:journals/corr/abs-2103-04234}, who find HS performance to be 3,500 tx/s on LAN without modifications such as only transmitting hashes~\cite{DBLP:journals/corr/abs-1906-05552}. Performance evaluations~\cite{DBLP:journals/corr/abs-1912-05241} of LibraBFT~\cite{baudet2019state} that uses Baseline HS, report throughput of around 500 tx/s.

\para{Batched HS}
For Batched HS without faults (see \Cref{fig:latency}, Batched HS lines), the maximum throughput we observe is 70,000 tx/s for a committee of 10 nodes, and lower (up to 50,000 tx/s) for a larger committee of 20. Latency before saturation is around 2 seconds. 
The almost 20x performance improvement compared to baseline HS is evidence that decoupling transaction dissemination from the critical path of consensus is the key to blockchain scalability. Performance decrease with the committee size similarly to other works~\cite{DBLP:journals/corr/abs-2103-04234}.

\para{\sysname HS}
The Nar\-whal-HS lines, show the through\-put-latency characteristics of using the \sysname mempool with HS, for different committee sizes, and 1 worker collocated with each validator. We observe that it achieves 140,000 tx/sec at a latency consistently below 2 seconds similar to the Batched HS latency. This validates the benefit of separating mempool from consensus, with mempool largely affecting the throughput and consensus affecting latency. The counter-intuitive fact that throughput increases with the committee size is due to the worker's implementation not using all resources (network, disk, CPU) optimally. Therefore, more validators and workers lead to increased multiplexing of resource use and higher performance for \sysname.

\para{\consname}
Finally, the \consname lines illustrate the latency-through\-put for various committee sizes, and 1 worker collocated with the primary per validator. We observe a stable latency at around 3 secs for all committee sizes as well as a peak throughput at 170,000 tx/sec for 50 validators. This is by far the most performant fully asynchronous consensus protocol (see \Cref{sec:related}) and only slightly worse than \sysname HS. The difference in reported latency compared to theory is that \first we not only put transactions in the blocks of the first round of a consensus instance but in all rounds, which increases the expected latency by $0.5$ rounds and \second synchronization is conservative, to not flood validators, hence it takes slightly longer for a validator to collect the full causal graph and start the total ordering.

\subsection{Scale-Out using many workers}
To evaluate the scalability of \sysname, we measure throughput with each worker and primary on a dedicated instance. For example, a validator running with 4 workers is implemented on 5 machines (4 workers + the primary). The workers are in the same data center as their primary, and validators are distributed over the 5 data centers over a WAN.

\begin{figure}[t]
\centering
\includegraphics[width=\columnwidth]{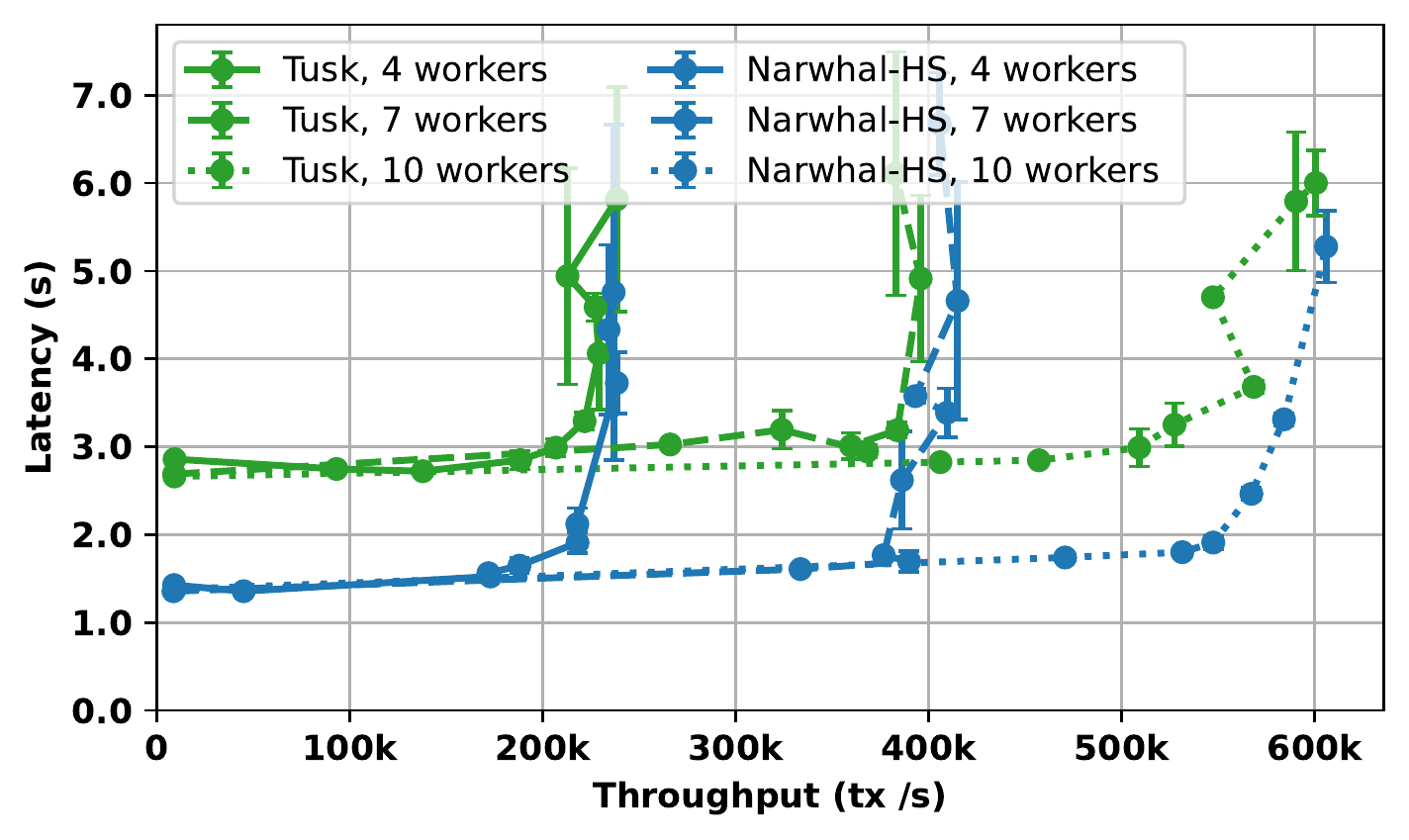}\\    

\includegraphics[width=\columnwidth]{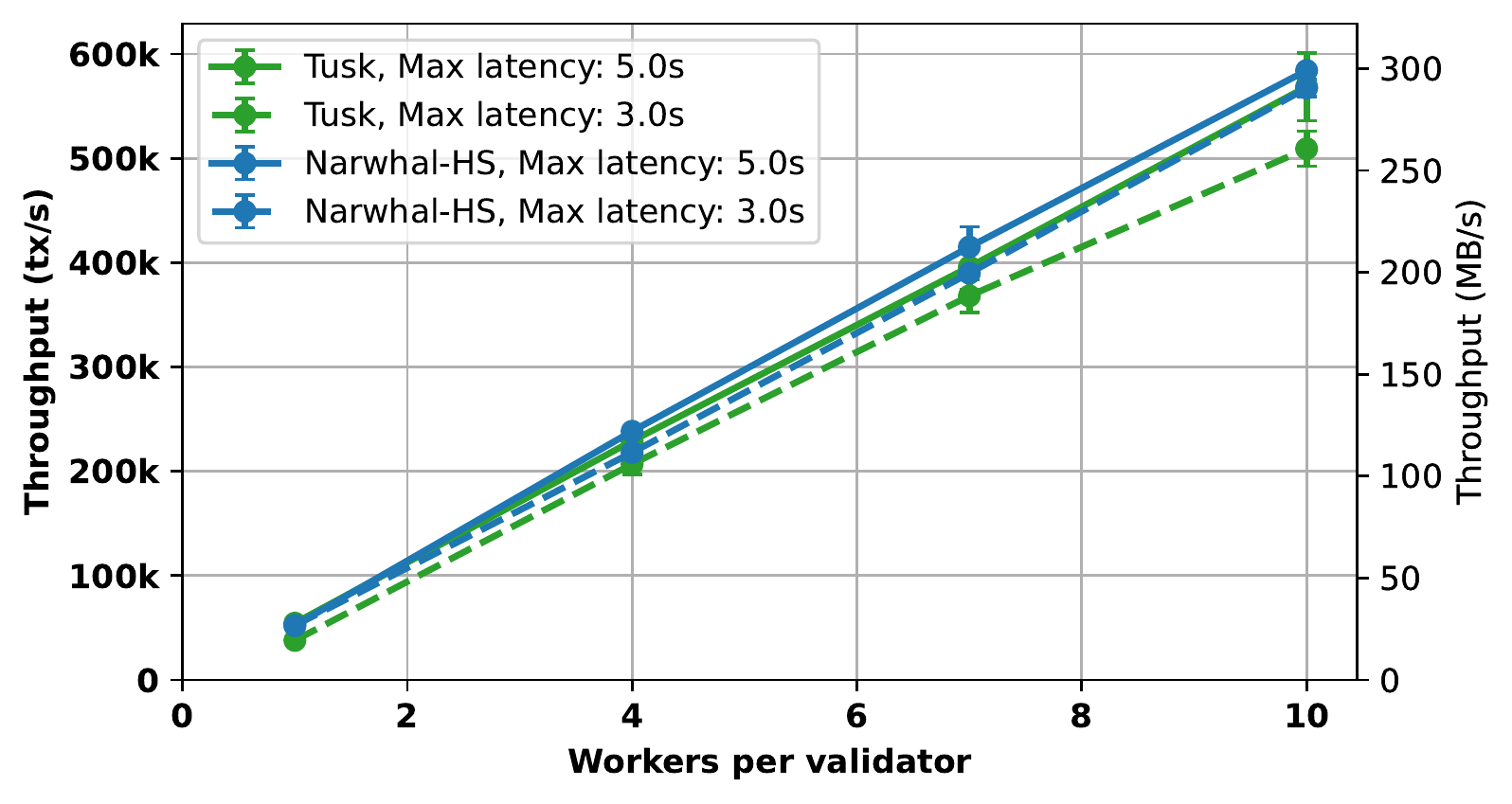}
   
    \caption{\consname and HS with \sysname latency-throughput graph for 4 validators and different number of workers. The transaction and batch sizes are respectively set to 512B and 1,000 transactions.}
    \label{fig:scalability-latency}
\end{figure}

The top \Cref{fig:scalability-latency} illustrates the latency-throughput graph of \sysname HS and \consname for a various number of workers per authority whereas the bottom \Cref{fig:scalability-latency} shows the maximum achievable throughput under various service level objectives (SLO). As expected, the deployments with a large number of workers saturate later while they all maintain the same latency, proving our claim that the primary is far from being saturated even with 10 workers concurrently serving hashes of batches. 
Additionally, on the SLO graph, we can see the linear scaling as the throughput is close to:
$$(\text{\#workers})*(\text{throughput~for~one~worker})$$




\subsection{Performance under Faults}

\Cref{fig:latency_faults} depicts the performance of all systems when a committee of 10 validators suffers 1 or 3 crash-faults (the maximum that can be tolerated). Both baseline and batched HotStuff suffer a massive degradation in throughput as well as a dramatic increase in latency. For three faults, baseline HotStuff throughput drops 5x (from a very low throughput of 800 tx/s to start with) and latency increases 40x compared to no faults; batched Hotstuff throughput drops 30x (from 70k tx/sec to 2.5k tx/sec) and latency increases 10x.

\begin{figure}[t]
    \centering
    \includegraphics[width=\columnwidth]{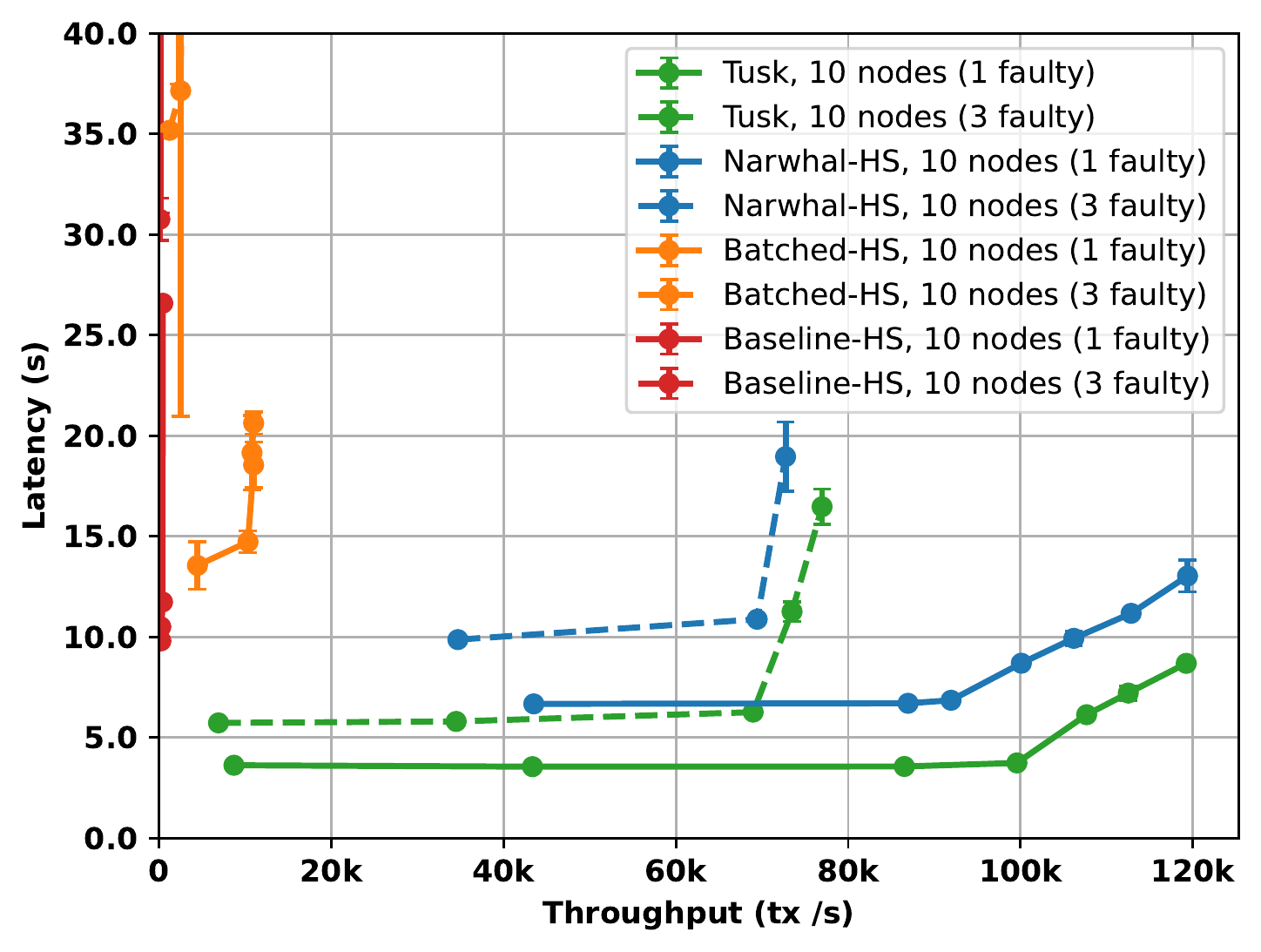}
    \caption{Comparative throughput-latency under faults. WAN measurements with 10 validators, using 1 worker collocated with the primary. One and three faults, 500KB max.\ block size and 512B transaction size.}
    \label{fig:latency_faults}
\end{figure}

In contrast, \consname and \sysname-HotStuff maintain a good level of throughput: the underlying \mempool design continues collecting and disseminating transactions despite the faults, and is not overly affected by the faulty validators. The reduction in throughput is in great part due to losing the capacity of faulty validators. As predicted by theory, \consname's latency is the least affected by faulty nodes, committing in less than 4 sec under 1 fault, and less than 6 sec under 3 faults. The increase in latency is due to the elected block being absent more often. \sysname-HotStuff exhibits a higher latency, but surprisingly lower than the baseline or batched variants, at less than 7 sec for 1 fault, and around 10 sec for 3 faults (compared to 35-40 sec for baseline or batched). We conjecture this is due to the very low throughput requirements placed on the protocol when combined with \sysname, as well as the fact that the first successful commit commits a large fraction of recent transactions thanks to the $2/3$-Causality.

\section{Related work \& Limitations} \label{sec:related}


\subsection{Performance} 


Widely-used blockchain systems such as Tendermint, provide 5k tx/sec~\cite{buchman2016tendermint}. However, performance under attack is much contested. Early work suggests that specially crafted attacks can degrade the performance of PBFT consensus systems massively~\cite{amir2010prime}, to such a lower performance point that liveness guarantees are meaningless. Recent work targeting PBFT in Hyperledger Fabric~\cite{DBLP:conf/eurosys/AndroulakiBBCCC18} corroborates these results~\cite{nguyen2019impact} and shows latency grows from a few seconds to a few hundred seconds, just through blocks being delayed. Han et al.~\cite{han2020performance} report similar dramatic performance degradation in cases of simple node crash failure for a number of quorum based systems namely Hyperledger Fabric, Ripple and Corda. In contrast, we demonstrate that \sysname combined with a traditional consensus mechanism, such as HotStuff maintains its throughput under attack, with increased latency. 

Mir-BFT \cite{DBLP:journals/corr/abs-1906-05552}, is the most performant variant of PBFT available. For transaction sizes of about 500B (similar to our benchmarks), the peak performance achieved on a WAN for 20 validators is around 80,000 tx/sec under 2 seconds -- a performance comparable to our baseline HotStuff with a batched mempool. Impressively, this throughput decreases only slowly for large committees up to 100 nodes (at 60,000 tx/sec). Faults lead to throughput dropping to zero for up to 50 seconds, and then operation resuming after a reconfiguration to exclude faulty nodes. \consname's single worker configuration for larger committees offers slightly higher performance (a bit less than 2x), but at double the latency. However, \consname with multiple workers allows 9x better throughput, and is much less sensitive to faults. 

We note that careful engineering of the mempool, and efficient transaction dissemination, seems crucial to achieving high-throughput consensus protocols. Recent work~\cite{DBLP:journals/corr/abs-2103-04234}, benchmarks crash-fault and Byzantine protocols on a LAN, yet observes orders of magnitude lower throughput than this work or Mir-BFT on WAN: 3,500 tx/sec for HotStuff and 500 tx/sec for PBFT. These results are comparable with the poor baseline we achieved when operating HotStuff with a naive best-effort broadcast mempool (see \Cref{fig:latency}). 

In the blockchain world, scale-out has come to mean sharding, but this only focuses in the case that every machine distrusts every other machine. However, in classic data-center setting there is a simpler scale-out solution before sharding, that of specializing machines, which we also use. In this paper we do not deal with sharding since our consensus algorithm can obviously interface with any sharded blockchain~\cite{DBLP:conf/sp/Kokoris-KogiasJ18,DBLP:conf/ndss/Al-BassamSBHD18,DBLP:conf/ccs/ZamaniM018,avarikioti2021divide}.

\subsection{Asynchronous Consensus \& \consname}



In the last 5 years the search of practical asynchronous consensus has captured the community~\cite{DBLP:conf/ccs/GuoL0XZ20, DBLP:conf/ccs/DuanRZ18, miller2016honey, spiegelman2019ace}, because of the high robustness guarantees it promises. 
The most performant one was Dumbo2 \cite{DBLP:conf/ccs/GuoL0XZ20} , which achieves a throughput of 5,000 tx/sec for an SLO of 5 seconds in a setting of 8 nodes in WAN with transaction size of 250B.
However, despite the significant improvement, the reported numbers do not realize the hope for a system that can support hundreds thousands of tx/sec
We believe that by showing a speedup of 20X over Dumbo2, \consname finally proves that asynchronous Byzantine consensus can be highly efficient and thus a strong candidate for future scalable deployed Blockchains. Additionally, Tusk manages to provide censorship-resistance at no extra cost, unlike the Honeybadger BFT variants that use threshold encryption.

The most closely related work to \consname is DAG-Rider~\cite{DBLP:journals/corr/abs-2102-08325}.
For the core algorithmic part we extend DAG-Rider: (i) we replace the classic reliable broadcast with our version described in \Cref{sec:qb-rbc}; (ii) change the commit rule to have a better common-case latency; and (iii) remove weak links to allow garbage collection.

\com{
which is a theoretical concurrent work on the Byzantine Atomic Broadcast problem. 
They use a DAG structure that resembles ours and have the same asymptotic worse-case analysis. 
Moreover, their theoretical security argument supports our safety claims.
As for liveness, \consname has a lower number of rounds in the common case due to an optimistic commit rule, but DAG-Rider has better worse-case round complexity (if counted in block of the DAG).
In addition,} DAG-Rider uses the notion of weak-links to achieve the eventual fairness property required by atomic broadcast, i.e., that any block broadcast by an honest party will be eventually committed. 
This make garbage collection impossible, as every block ever received has to be stored in-case it is pointed by a weak link. Thus, the strong notion of fairness of DAG-Rider is unimplementable within finite storage and purely theoretical. In \consname we forbid weak links. We instead re-inject transactions of uncommitted garbage collected blocks into new blocks in subsequent rounds. Thus, instead of block-level fairness this achieves the more realistic metric of transaction-level fairness which is sufficient for censorship-resistance of transactions.

All in all, we believe that DAG-Rider further validates the design of \sysname, since it would take less than 200 LOC to implement DAG-Rider over \sysname.

\subsection{DAG-based Communication  \& \sysname}

The directed acyclic graph (DAG) data-structure as a substrate for capturing the communication of secure distributed systems in the context of Blockchains has been proposed multiple times. 
The layered structure of our DAG has also been proposed in the crash fault setting by Ford~\cite{tlc}, it is however, embedded in the consensus protocol and does not  leverage it for batching and pipelining.
%
Hashgraph~\cite{baird2016swirlds} embeds an asynchronous consensus
mechanism onto a DAG of degree two, without a layered Threshold Clock
structure. As a result the logic for when older blocks are no more
needed is complex and unclear, and garbage collecting them difficult --
leading to potentially unbounded state to decide future blocks. 
In addition, they use local coins for randomness, which can potentially
lead to exponential latency. 
Blockmania~\cite{blockmania} embeds a single-decision variant
of PBFT~\cite{CastroL02} into a non layered DAG
leading to a partially synchronous system, with challenges when it
comes to garbage collection -- as any past blocks may be required for
long range decisions in the future. Neither of these protocols use a clear
decomposition between lower level sharded availability, and a higher
level consensus protocol as \sysname offers, and thus do not
scale out or offer clear ways to garbage collect old blocks.

\subsection{Limitations}

A limitation of any reactive asynchronous protocol, including \sysname and \consname, 
is that slow authorities are indistinguishable from faulty ones, and as a result the 
protocol proceeds without them. This creates issues around fairness and incentives, 
since perfectly correct, but geographically distant authorities may only be able to 
commit transactions submitted to them much later. This is a generic limitation of such protocols, 
and we leave the definition and implementation of fairness mechanisms to future work.
Nevertheless, we note that we get $1/2$-Chain Quality or at least 50\% of all blocks are made by honest parties, which is to the best of our knowledge the highest number any existing proposal achieves. 
Additionally, from a theoretical perspective this can be argued to be 100\% as all transactions will eventually be delivered (through re-injection) and eventually committed. However treating Chain Quality as a liveness property is quite meaningless in practice as clients don't have an infinity to wait. 

Further, \sysname relies on clients to re-submit a transaction if it is not sequenced 
in time, due to the leader being faulty. An expensive alternative is to require clients 
to submit a transaction to $f+1$ authorities, but this would divide the bandwidth of 
authorities by $O(n)$. This is too high a price to pay, since a client submitting a 
transaction to a fixed $k$ number of authorities has a probability of including a 
correct one that grows very quickly in $k$, at the cost of only $O(1)$ overhead. 
Notably, Mir-BFT uses an interesting transaction de-duplication technique based on hashing which we believe is directly applicable to \sysname in case such a feature is needed.
Ultimately, we relegate this choice to system designers using \sysname.

Traditional (monolithic) systems store all transactions data on the machine running consensus. \sysname however distributes transaction data among a number of worker machines. Despite, \sysname's approach allows to scale storage by adding more workers, it places a new burden on the execution engine. The execution engine now needs to locate and retriever transaction data from the various workers. Although \sysname's certificates irrevocable indicate which worker holds the transaction data, retrieving them requires an additional protocol and may introduce extra latency. Light clients face a similar issue, their design needs to adapt to locate and track transaction data across workers. 

Further, the high throughput of \sysname and \consname is only useful is there exist a transaction engine capable to match their speed. Building an efficient execution engine for \sysname is deferred to future work.

\balance

\section{Conclusion} \label{sec:conclusion}



We experimentally demonstrated the power of \sysname and \consname.
\sysname is an advanced mempool enabling Hotstuff to achieve throughput of $130,000$ tx/sec with under $2$ seconds latency, in a deployment of $50$ geographically distributed single-machine validators.
Additionally, \sysname enables any quorum-based blockchain protocol to maintain its throughput within periods of asynchrony or
faults, as long as the consensus layer is eventually live.
\consname leverages the structure of \sysname to achieve a throughput of $160,000$ TPS with about $3$ seconds latency. The scale-out design allows this throughput to increase to hundreds of thousands TPS without impact on latency.

In one sentence, \sysname and \consname conclusively prove that the main cost of large-scale blockchain protocols \emph{is not} consensus but the reliable transaction dissemination. Yet, dissemination alone, without global sequencing, is an embarrassingly parallelizable function, as we show with the scale-out design of \sysname. 

Our work supports a rethinking in how distributed ledgers and SMR systems are architected, towards pairing a mempool, like \sysname, to ensure high-throughput even under faults and asynchrony, with a consensus mechanism to achieve low-latency for fixed-size messages. \consname demonstrates that there exists a zero-message overhead consensus for \sysname, secure under full asynchrony.
As a result, quorum-based blockchains can scale to potentially millions of transactions per second through scale-out for payments or to build generic reliable systems through state machine replication and smart contracts.

\ifdefined\cameraReady
\section*{Acknowledgments}
The majority of this work has beend done when the authors were part of the Novi team at Facebook.
We also thank the Novi Research and Engineering teams for valuable
feedback, and in particular Mathieu Baudet, Andrey Chursin, Zekun Li, and Dahlia Malkhi for early discussions that shaped this work.
\fi

\bibliographystyle{ACM-Reference-Format}
\bibliography{references}


\begin{thebibliography}{40}


\ifx \showCODEN    \undefined \def \showCODEN     #1{\unskip}     \fi
\ifx \showDOI      \undefined \def \showDOI       #1{#1}\fi
\ifx \showISBNx    \undefined \def \showISBNx     #1{\unskip}     \fi
\ifx \showISBNxiii \undefined \def \showISBNxiii  #1{\unskip}     \fi
\ifx \showISSN     \undefined \def \showISSN      #1{\unskip}     \fi
\ifx \showLCCN     \undefined \def \showLCCN      #1{\unskip}     \fi
\ifx \shownote     \undefined \def \shownote      #1{#1}          \fi
\ifx \showarticletitle \undefined \def \showarticletitle #1{#1}   \fi
\ifx \showURL      \undefined \def \showURL       {\relax}        \fi
\providecommand\bibfield[2]{#2}
\providecommand\bibinfo[2]{#2}
\providecommand\natexlab[1]{#1}
\providecommand\showeprint[2][]{arXiv:#2}

\bibitem[\protect\citeauthoryear{Abadi and Faleiro}{Abadi and Faleiro}{2018}]%
        {abadi2018overview}
\bibfield{author}{\bibinfo{person}{Daniel~J Abadi} {and}
  \bibinfo{person}{Jose~M Faleiro}.} \bibinfo{year}{2018}\natexlab{}.
\newblock \showarticletitle{An overview of deterministic database systems}.
\newblock \bibinfo{journal}{\emph{Commun. ACM}} \bibinfo{volume}{61},
  \bibinfo{number}{9} (\bibinfo{year}{2018}), \bibinfo{pages}{78--88}.
\newblock


\bibitem[\protect\citeauthoryear{Abd{-}El{-}Malek, Ganger, Goodson, Reiter, and
  Wylie}{Abd{-}El{-}Malek et~al\mbox{.}}{2005}]%
        {DBLP:conf/sosp/Abd-El-MalekGGRW05}
\bibfield{author}{\bibinfo{person}{Michael Abd{-}El{-}Malek},
  \bibinfo{person}{Gregory~R. Ganger}, \bibinfo{person}{Garth~R. Goodson},
  \bibinfo{person}{Michael~K. Reiter}, {and} \bibinfo{person}{Jay~J. Wylie}.}
  \bibinfo{year}{2005}\natexlab{}.
\newblock \showarticletitle{Fault-scalable Byzantine fault-tolerant services}.
  In \bibinfo{booktitle}{\emph{Proceedings of the 20th {ACM} Symposium on
  Operating Systems Principles 2005, {SOSP} 2005, Brighton, UK, October 23-26,
  2005}}, \bibfield{editor}{\bibinfo{person}{Andrew Herbert} {and}
  \bibinfo{person}{Kenneth~P. Birman}} (Eds.). \bibinfo{publisher}{{ACM}},
  \bibinfo{pages}{59--74}.
\newblock


\bibitem[\protect\citeauthoryear{Abraham, Malkhi, and Spiegelman}{Abraham
  et~al\mbox{.}}{2019}]%
        {DBLP:conf/podc/AbrahamMS19}
\bibfield{author}{\bibinfo{person}{Ittai Abraham}, \bibinfo{person}{Dahlia
  Malkhi}, {and} \bibinfo{person}{Alexander Spiegelman}.}
  \bibinfo{year}{2019}\natexlab{}.
\newblock \showarticletitle{Asymptotically Optimal Validated Asynchronous
  Byzantine Agreement}. In \bibinfo{booktitle}{\emph{Proceedings of the 2019
  {ACM} Symposium on Principles of Distributed Computing, {PODC} 2019, Toronto,
  ON, Canada, July 29 - August 2, 2019}},
  \bibfield{editor}{\bibinfo{person}{Peter Robinson} {and}
  \bibinfo{person}{Faith Ellen}} (Eds.). \bibinfo{publisher}{{ACM}},
  \bibinfo{pages}{337--346}.
\newblock


\bibitem[\protect\citeauthoryear{Al{-}Bassam, Sonnino, Bano, Hrycyszyn, and
  Danezis}{Al{-}Bassam et~al\mbox{.}}{2018}]%
        {DBLP:conf/ndss/Al-BassamSBHD18}
\bibfield{author}{\bibinfo{person}{Mustafa Al{-}Bassam},
  \bibinfo{person}{Alberto Sonnino}, \bibinfo{person}{Shehar Bano},
  \bibinfo{person}{Dave Hrycyszyn}, {and} \bibinfo{person}{George Danezis}.}
  \bibinfo{year}{2018}\natexlab{}.
\newblock \showarticletitle{Chainspace: {A} Sharded Smart Contracts Platform}.
  In \bibinfo{booktitle}{\emph{25th Annual Network and Distributed System
  Security Symposium, {NDSS} 2018, San Diego, California, USA, February 18-21,
  2018}}. \bibinfo{publisher}{The Internet Society}.
\newblock


\bibitem[\protect\citeauthoryear{Alqahtani and Demirbas}{Alqahtani and
  Demirbas}{2021}]%
        {DBLP:journals/corr/abs-2103-04234}
\bibfield{author}{\bibinfo{person}{Salem Alqahtani} {and}
  \bibinfo{person}{Murat Demirbas}.} \bibinfo{year}{2021}\natexlab{}.
\newblock \showarticletitle{Bottlenecks in Blockchain Consensus Protocols}.
\newblock \bibinfo{journal}{\emph{CoRR}}  \bibinfo{volume}{abs/2103.04234}
  (\bibinfo{year}{2021}).
\newblock


\bibitem[\protect\citeauthoryear{Amir, Coan, Kirsch, and Lane}{Amir
  et~al\mbox{.}}{2010}]%
        {amir2010prime}
\bibfield{author}{\bibinfo{person}{Yair Amir}, \bibinfo{person}{Brian Coan},
  \bibinfo{person}{Jonathan Kirsch}, {and} \bibinfo{person}{John Lane}.}
  \bibinfo{year}{2010}\natexlab{}.
\newblock \showarticletitle{Prime: Byzantine replication under attack}.
\newblock \bibinfo{journal}{\emph{IEEE transactions on dependable and secure
  computing}} \bibinfo{volume}{8}, \bibinfo{number}{4} (\bibinfo{year}{2010}),
  \bibinfo{pages}{564--577}.
\newblock


\bibitem[\protect\citeauthoryear{Androulaki, Barger, Bortnikov, Cachin,
  Christidis, Caro, Enyeart, Ferris, Laventman, Manevich, Muralidharan, Murthy,
  Nguyen, Sethi, Singh, Smith, Sorniotti, Stathakopoulou, Vukolic, Cocco, and
  Yellick}{Androulaki et~al\mbox{.}}{2018}]%
        {DBLP:conf/eurosys/AndroulakiBBCCC18}
\bibfield{author}{\bibinfo{person}{Elli Androulaki}, \bibinfo{person}{Artem
  Barger}, \bibinfo{person}{Vita Bortnikov}, \bibinfo{person}{Christian
  Cachin}, \bibinfo{person}{Konstantinos Christidis},
  \bibinfo{person}{Angelo~De Caro}, \bibinfo{person}{David Enyeart},
  \bibinfo{person}{Christopher Ferris}, \bibinfo{person}{Gennady Laventman},
  \bibinfo{person}{Yacov Manevich}, \bibinfo{person}{Srinivasan Muralidharan},
  \bibinfo{person}{Chet Murthy}, \bibinfo{person}{Binh Nguyen},
  \bibinfo{person}{Manish Sethi}, \bibinfo{person}{Gari Singh},
  \bibinfo{person}{Keith Smith}, \bibinfo{person}{Alessandro Sorniotti},
  \bibinfo{person}{Chrysoula Stathakopoulou}, \bibinfo{person}{Marko Vukolic},
  \bibinfo{person}{Sharon~Weed Cocco}, {and} \bibinfo{person}{Jason Yellick}.}
  \bibinfo{year}{2018}\natexlab{}.
\newblock \showarticletitle{Hyperledger fabric: a distributed operating system
  for permissioned blockchains}. In \bibinfo{booktitle}{\emph{Proceedings of
  the Thirteenth EuroSys Conference, EuroSys 2018, Porto, Portugal, April
  23-26, 2018}}, \bibfield{editor}{\bibinfo{person}{Rui Oliveira},
  \bibinfo{person}{Pascal Felber}, {and} \bibinfo{person}{Y.~Charlie Hu}}
  (Eds.). \bibinfo{publisher}{{ACM}}, \bibinfo{pages}{30:1--30:15}.
\newblock


\bibitem[\protect\citeauthoryear{Avarikioti, Kokoris-Kogias, and
  Wattenhofer}{Avarikioti et~al\mbox{.}}{2021}]%
        {avarikioti2021divide}
\bibfield{author}{\bibinfo{person}{Georgia Avarikioti},
  \bibinfo{person}{Eleftherios Kokoris-Kogias}, {and} \bibinfo{person}{Roger
  Wattenhofer}.} \bibinfo{year}{2021}\natexlab{}.
\newblock \bibinfo{title}{Divide and Scale: Formalization of Distributed Ledger
  Sharding Protocols}.
\newblock
\newblock
\showeprint[arxiv]{1910.10434}~[cs.DC]


\bibitem[\protect\citeauthoryear{Bagaria, Kannan, Tse, Fanti, and
  Viswanath}{Bagaria et~al\mbox{.}}{2019}]%
        {prism}
\bibfield{author}{\bibinfo{person}{Vivek Bagaria}, \bibinfo{person}{Sreeram
  Kannan}, \bibinfo{person}{David Tse}, \bibinfo{person}{Giulia Fanti}, {and}
  \bibinfo{person}{Pramod Viswanath}.} \bibinfo{year}{2019}\natexlab{}.
\newblock \showarticletitle{Prism: Deconstructing the blockchain to approach
  physical limits}. In \bibinfo{booktitle}{\emph{Proceedings of the 2019 ACM
  SIGSAC Conference on Computer and Communications Security}}.
  \bibinfo{pages}{585--602}.
\newblock


\bibitem[\protect\citeauthoryear{Baird}{Baird}{2016}]%
        {baird2016swirlds}
\bibfield{author}{\bibinfo{person}{Leemon Baird}.}
  \bibinfo{year}{2016}\natexlab{}.
\newblock \showarticletitle{The swirlds hashgraph consensus algorithm: Fair,
  fast, byzantine fault tolerance}.
\newblock \bibinfo{journal}{\emph{Swirlds Tech Reports SWIRLDS-TR-2016-01,
  Tech. Rep}} (\bibinfo{year}{2016}).
\newblock


\bibitem[\protect\citeauthoryear{Bano, Sonnino, Chursin, Perelman, and
  Malkhi}{Bano et~al\mbox{.}}{2020}]%
        {twins}
\bibfield{author}{\bibinfo{person}{Shehar Bano}, \bibinfo{person}{Alberto
  Sonnino}, \bibinfo{person}{Andrey Chursin}, \bibinfo{person}{Dmitri
  Perelman}, {and} \bibinfo{person}{Dahlia Malkhi}.}
  \bibinfo{year}{2020}\natexlab{}.
\newblock \showarticletitle{Twins: White-Glove Approach for BFT Testing}.
\newblock \bibinfo{journal}{\emph{arXiv preprint arXiv:2004.10617}}
  (\bibinfo{year}{2020}).
\newblock


\bibitem[\protect\citeauthoryear{Baudet, Ching, Chursin, Danezis, Garillot, Li,
  Malkhi, Naor, Perelman, and Sonnino}{Baudet et~al\mbox{.}}{2019}]%
        {baudet2019state}
\bibfield{author}{\bibinfo{person}{Mathieu Baudet}, \bibinfo{person}{Avery
  Ching}, \bibinfo{person}{Andrey Chursin}, \bibinfo{person}{George Danezis},
  \bibinfo{person}{Fran{\c{c}}ois Garillot}, \bibinfo{person}{Zekun Li},
  \bibinfo{person}{Dahlia Malkhi}, \bibinfo{person}{Oded Naor},
  \bibinfo{person}{Dmitri Perelman}, {and} \bibinfo{person}{Alberto Sonnino}.}
  \bibinfo{year}{2019}\natexlab{}.
\newblock \showarticletitle{State machine replication in the Libra blockchain}.
\newblock \bibinfo{journal}{\emph{The Libra Assn., Tech. Rep}}
  (\bibinfo{year}{2019}).
\newblock


\bibitem[\protect\citeauthoryear{Bessani, Sousa, and Alchieri}{Bessani
  et~al\mbox{.}}{2014}]%
        {bessani2014state}
\bibfield{author}{\bibinfo{person}{Alysson Bessani}, \bibinfo{person}{Joao
  Sousa}, {and} \bibinfo{person}{Eduardo~EP Alchieri}.}
  \bibinfo{year}{2014}\natexlab{}.
\newblock \showarticletitle{State machine replication for the masses with
  BFT-SMART}. In \bibinfo{booktitle}{\emph{2014 44th Annual IEEE/IFIP
  International Conference on Dependable Systems and Networks}}. IEEE,
  \bibinfo{pages}{355--362}.
\newblock


\bibitem[\protect\citeauthoryear{Boneh, Lynn, and Shacham}{Boneh
  et~al\mbox{.}}{2004}]%
        {DBLP:journals/joc/BonehLS04}
\bibfield{author}{\bibinfo{person}{Dan Boneh}, \bibinfo{person}{Ben Lynn},
  {and} \bibinfo{person}{Hovav Shacham}.} \bibinfo{year}{2004}\natexlab{}.
\newblock \showarticletitle{Short Signatures from the Weil Pairing}.
\newblock \bibinfo{journal}{\emph{J. Cryptol.}} \bibinfo{volume}{17},
  \bibinfo{number}{4} (\bibinfo{year}{2004}), \bibinfo{pages}{297--319}.
\newblock


\bibitem[\protect\citeauthoryear{Bracha}{Bracha}{1987}]%
        {bracha1987asynchronous}
\bibfield{author}{\bibinfo{person}{Gabriel Bracha}.}
  \bibinfo{year}{1987}\natexlab{}.
\newblock \showarticletitle{Asynchronous Byzantine agreement protocols}.
\newblock \bibinfo{journal}{\emph{Information and Computation}}
  \bibinfo{volume}{75}, \bibinfo{number}{2} (\bibinfo{year}{1987}),
  \bibinfo{pages}{130--143}.
\newblock


\bibitem[\protect\citeauthoryear{Bracha and Toueg}{Bracha and Toueg}{1985}]%
        {bracha1985asynchronous}
\bibfield{author}{\bibinfo{person}{Gabriel Bracha} {and} \bibinfo{person}{Sam
  Toueg}.} \bibinfo{year}{1985}\natexlab{}.
\newblock \showarticletitle{Asynchronous consensus and broadcast protocols}.
\newblock \bibinfo{journal}{\emph{Journal of the ACM (JACM)}}
  \bibinfo{volume}{32}, \bibinfo{number}{4} (\bibinfo{year}{1985}),
  \bibinfo{pages}{824--840}.
\newblock


\bibitem[\protect\citeauthoryear{Buchman}{Buchman}{2016}]%
        {buchman2016tendermint}
\bibfield{author}{\bibinfo{person}{Ethan Buchman}.}
  \bibinfo{year}{2016}\natexlab{}.
\newblock \bibinfo{title}{{Tendermint: Byzantine fault tolerance in the age of
  blockchains}}.
\newblock
\newblock


\bibitem[\protect\citeauthoryear{Cachin, Guerraoui, and Rodrigues}{Cachin
  et~al\mbox{.}}{2011}]%
        {cachin2011introduction}
\bibfield{author}{\bibinfo{person}{Christian Cachin}, \bibinfo{person}{Rachid
  Guerraoui}, {and} \bibinfo{person}{Lu{\'\i}s Rodrigues}.}
  \bibinfo{year}{2011}\natexlab{}.
\newblock \bibinfo{booktitle}{\emph{Introduction to reliable and secure
  distributed programming}}.
\newblock \bibinfo{publisher}{Springer Science \& Business Media}.
\newblock


\bibitem[\protect\citeauthoryear{Castro and Liskov}{Castro and Liskov}{2002}]%
        {CastroL02}
\bibfield{author}{\bibinfo{person}{Miguel Castro} {and}
  \bibinfo{person}{Barbara Liskov}.} \bibinfo{year}{2002}\natexlab{}.
\newblock \showarticletitle{Practical byzantine fault tolerance and proactive
  recovery}.
\newblock \bibinfo{journal}{\emph{{ACM} Trans. Comput. Syst.}}
  \bibinfo{volume}{20}, \bibinfo{number}{4} (\bibinfo{year}{2002}),
  \bibinfo{pages}{398--461}.
\newblock
\urldef\tempurl%
\url{https://doi.org/10.1145/571637.571640}
\showDOI{\tempurl}


\bibitem[\protect\citeauthoryear{Danezis and Hrycyszyn}{Danezis and
  Hrycyszyn}{2018}]%
        {blockmania}
\bibfield{author}{\bibinfo{person}{George Danezis} {and} \bibinfo{person}{David
  Hrycyszyn}.} \bibinfo{year}{2018}\natexlab{}.
\newblock \showarticletitle{Blockmania: from block DAGs to consensus}.
\newblock \bibinfo{journal}{\emph{arXiv preprint arXiv:1809.01620}}
  (\bibinfo{year}{2018}).
\newblock


\bibitem[\protect\citeauthoryear{Dolev, Fischer, Fowler, Lynch, and
  Strong}{Dolev et~al\mbox{.}}{1982}]%
        {dolev1982efficient}
\bibfield{author}{\bibinfo{person}{Danny Dolev}, \bibinfo{person}{Michael~J
  Fischer}, \bibinfo{person}{Rob Fowler}, \bibinfo{person}{Nancy~A Lynch},
  {and} \bibinfo{person}{H~Raymond Strong}.} \bibinfo{year}{1982}\natexlab{}.
\newblock \showarticletitle{An efficient algorithm for Byzantine agreement
  without authentication}.
\newblock \bibinfo{journal}{\emph{Information and Control}}
  \bibinfo{volume}{52}, \bibinfo{number}{3} (\bibinfo{year}{1982}),
  \bibinfo{pages}{257--274}.
\newblock


\bibitem[\protect\citeauthoryear{Dolev and Reischuk}{Dolev and
  Reischuk}{1985}]%
        {dolev1985bounds}
\bibfield{author}{\bibinfo{person}{Danny Dolev} {and} \bibinfo{person}{Rudiger
  Reischuk}.} \bibinfo{year}{1985}\natexlab{}.
\newblock \showarticletitle{Bounds on information exchange for Byzantine
  agreement}.
\newblock \bibinfo{journal}{\emph{JACM}} (\bibinfo{year}{1985}).
\newblock


\bibitem[\protect\citeauthoryear{Duan, Reiter, and Zhang}{Duan
  et~al\mbox{.}}{2018}]%
        {DBLP:conf/ccs/DuanRZ18}
\bibfield{author}{\bibinfo{person}{Sisi Duan}, \bibinfo{person}{Michael~K.
  Reiter}, {and} \bibinfo{person}{Haibin Zhang}.}
  \bibinfo{year}{2018}\natexlab{}.
\newblock \showarticletitle{{BEAT:} Asynchronous {BFT} Made Practical}. In
  \bibinfo{booktitle}{\emph{Proceedings of the 2018 {ACM} {SIGSAC} Conference
  on Computer and Communications Security, {CCS} 2018, Toronto, ON, Canada,
  October 15-19, 2018}}, \bibfield{editor}{\bibinfo{person}{David Lie},
  \bibinfo{person}{Mohammad Mannan}, \bibinfo{person}{Michael Backes}, {and}
  \bibinfo{person}{XiaoFeng Wang}} (Eds.). \bibinfo{publisher}{{ACM}},
  \bibinfo{pages}{2028--2041}.
\newblock


\bibitem[\protect\citeauthoryear{Ford}{Ford}{2019}]%
        {tlc}
\bibfield{author}{\bibinfo{person}{Bryan Ford}.}
  \bibinfo{year}{2019}\natexlab{}.
\newblock \showarticletitle{Threshold logical clocks for asynchronous
  distributed coordination and consensus}.
\newblock \bibinfo{journal}{\emph{arXiv preprint arXiv:1907.07010}}
  (\bibinfo{year}{2019}).
\newblock


\bibitem[\protect\citeauthoryear{Garay, Kiayias, and Leonardos}{Garay
  et~al\mbox{.}}{2015}]%
        {DBLP:conf/eurocrypt/GarayKL15}
\bibfield{author}{\bibinfo{person}{Juan~A. Garay}, \bibinfo{person}{Aggelos
  Kiayias}, {and} \bibinfo{person}{Nikos Leonardos}.}
  \bibinfo{year}{2015}\natexlab{}.
\newblock \showarticletitle{The Bitcoin Backbone Protocol: Analysis and
  Applications}. In \bibinfo{booktitle}{\emph{Advances in Cryptology -
  {EUROCRYPT} 2015 - 34th Annual International Conference on the Theory and
  Applications of Cryptographic Techniques, Sofia, Bulgaria, April 26-30, 2015,
  Proceedings, Part {II}}} \emph{(\bibinfo{series}{Lecture Notes in Computer
  Science}, Vol.~\bibinfo{volume}{9057})},
  \bibfield{editor}{\bibinfo{person}{Elisabeth Oswald} {and}
  \bibinfo{person}{Marc Fischlin}} (Eds.). \bibinfo{publisher}{Springer},
  \bibinfo{pages}{281--310}.
\newblock
\urldef\tempurl%
\url{https://doi.org/10.1007/978-3-662-46803-6\_10}
\showDOI{\tempurl}


\bibitem[\protect\citeauthoryear{Guo, Lu, Tang, Xu, and Zhang}{Guo
  et~al\mbox{.}}{2020}]%
        {DBLP:conf/ccs/GuoL0XZ20}
\bibfield{author}{\bibinfo{person}{Bingyong Guo}, \bibinfo{person}{Zhenliang
  Lu}, \bibinfo{person}{Qiang Tang}, \bibinfo{person}{Jing Xu}, {and}
  \bibinfo{person}{Zhenfeng Zhang}.} \bibinfo{year}{2020}\natexlab{}.
\newblock \showarticletitle{Dumbo: Faster Asynchronous {BFT} Protocols}. In
  \bibinfo{booktitle}{\emph{{CCS} '20: 2020 {ACM} {SIGSAC} Conference on
  Computer and Communications Security, Virtual Event, USA, November 9-13,
  2020}}, \bibfield{editor}{\bibinfo{person}{Jay Ligatti},
  \bibinfo{person}{Xinming Ou}, \bibinfo{person}{Jonathan Katz}, {and}
  \bibinfo{person}{Giovanni Vigna}} (Eds.). \bibinfo{publisher}{{ACM}},
  \bibinfo{pages}{803--818}.
\newblock


\bibitem[\protect\citeauthoryear{Han, Shapiro, Gramoli, and Xu}{Han
  et~al\mbox{.}}{2020}]%
        {han2020performance}
\bibfield{author}{\bibinfo{person}{Runchao Han}, \bibinfo{person}{Gary
  Shapiro}, \bibinfo{person}{Vincent Gramoli}, {and} \bibinfo{person}{Xiwei
  Xu}.} \bibinfo{year}{2020}\natexlab{}.
\newblock \showarticletitle{On the performance of distributed ledgers for
  internet of things}.
\newblock \bibinfo{journal}{\emph{Internet of Things}}  \bibinfo{volume}{10}
  (\bibinfo{year}{2020}), \bibinfo{pages}{100087}.
\newblock


\bibitem[\protect\citeauthoryear{Keidar, Kokoris{-}Kogias, Naor, and
  Spiegelman}{Keidar et~al\mbox{.}}{2021}]%
        {DBLP:journals/corr/abs-2102-08325}
\bibfield{author}{\bibinfo{person}{Idit Keidar}, \bibinfo{person}{Eleftherios
  Kokoris{-}Kogias}, \bibinfo{person}{Oded Naor}, {and}
  \bibinfo{person}{Alexander Spiegelman}.} \bibinfo{year}{2021}\natexlab{}.
\newblock \showarticletitle{All You Need is {DAG}}. In
  \bibinfo{booktitle}{\emph{Proceedings of the 40th Symposium on Principles of
  Distributed Computing}} (Virtual Event, Italy) \emph{(\bibinfo{series}{PODC
  '21})}. \bibinfo{publisher}{Association for Computing Machinery},
  \bibinfo{address}{New York, NY, USA}.
\newblock


\bibitem[\protect\citeauthoryear{Kogias, Jovanovic, Gailly, Khoffi, Gasser, and
  Ford}{Kogias et~al\mbox{.}}{2016}]%
        {DBLP:journals/corr/Kokoris-KogiasJ16}
\bibfield{author}{\bibinfo{person}{Eleftherios~Kokoris Kogias},
  \bibinfo{person}{Philipp Jovanovic}, \bibinfo{person}{Nicolas Gailly},
  \bibinfo{person}{Ismail Khoffi}, \bibinfo{person}{Linus Gasser}, {and}
  \bibinfo{person}{Bryan Ford}.} \bibinfo{year}{2016}\natexlab{}.
\newblock \showarticletitle{Enhancing Bitcoin Security and Performance with
  Strong Consistency via Collective Signing}. In \bibinfo{booktitle}{\emph{25th
  {USENIX} Security Symposium ({USENIX} Security 16)}}.
  \bibinfo{publisher}{{USENIX} Association}, \bibinfo{address}{Austin, TX},
  \bibinfo{pages}{279--296}.
\newblock
\showISBNx{978-1-931971-32-4}
\urldef\tempurl%
\url{https://www.usenix.org/conference/usenixsecurity16/technical-sessions/presentation/kogias}
\showURL{%
\tempurl}


\bibitem[\protect\citeauthoryear{Kokoris{-}Kogias, Jovanovic, Gasser, Gailly,
  Syta, and Ford}{Kokoris{-}Kogias et~al\mbox{.}}{2018}]%
        {DBLP:conf/sp/Kokoris-KogiasJ18}
\bibfield{author}{\bibinfo{person}{Eleftherios Kokoris{-}Kogias},
  \bibinfo{person}{Philipp Jovanovic}, \bibinfo{person}{Linus Gasser},
  \bibinfo{person}{Nicolas Gailly}, \bibinfo{person}{Ewa Syta}, {and}
  \bibinfo{person}{Bryan Ford}.} \bibinfo{year}{2018}\natexlab{}.
\newblock \showarticletitle{OmniLedger: {A} Secure, Scale-Out, Decentralized
  Ledger via Sharding}. In \bibinfo{booktitle}{\emph{2018 {IEEE} Symposium on
  Security and Privacy, {SP} 2018, Proceedings, 21-23 May 2018, San Francisco,
  California, {USA}}}. \bibinfo{publisher}{{IEEE} Computer Society},
  \bibinfo{pages}{583--598}.
\newblock


\bibitem[\protect\citeauthoryear{Kokoris{-}Kogias, Malkhi, and
  Spiegelman}{Kokoris{-}Kogias et~al\mbox{.}}{2020}]%
        {DBLP:conf/ccs/Kokoris-KogiasM20}
\bibfield{author}{\bibinfo{person}{Eleftherios Kokoris{-}Kogias},
  \bibinfo{person}{Dahlia Malkhi}, {and} \bibinfo{person}{Alexander
  Spiegelman}.} \bibinfo{year}{2020}\natexlab{}.
\newblock \showarticletitle{Asynchronous Distributed Key Generation for
  Computationally-Secure Randomness, Consensus, and Threshold Signatures}. In
  \bibinfo{booktitle}{\emph{{CCS} '20: 2020 {ACM} {SIGSAC} Conference on
  Computer and Communications Security, Virtual Event, USA, November 9-13,
  2020}}, \bibfield{editor}{\bibinfo{person}{Jay Ligatti},
  \bibinfo{person}{Xinming Ou}, \bibinfo{person}{Jonathan Katz}, {and}
  \bibinfo{person}{Giovanni Vigna}} (Eds.). \bibinfo{publisher}{{ACM}},
  \bibinfo{pages}{1751--1767}.
\newblock


\bibitem[\protect\citeauthoryear{Miller, Xia, Croman, Shi, and Song}{Miller
  et~al\mbox{.}}{2016}]%
        {miller2016honey}
\bibfield{author}{\bibinfo{person}{Andrew Miller}, \bibinfo{person}{Yu Xia},
  \bibinfo{person}{Kyle Croman}, \bibinfo{person}{Elaine Shi}, {and}
  \bibinfo{person}{Dawn Song}.} \bibinfo{year}{2016}\natexlab{}.
\newblock \showarticletitle{The honey badger of BFT protocols}. In
  \bibinfo{booktitle}{\emph{Proceedings of the 2016 ACM SIGSAC Conference on
  Computer and Communications Security}}. \bibinfo{pages}{31--42}.
\newblock


\bibitem[\protect\citeauthoryear{Nakamoto}{Nakamoto}{2008}]%
        {nakamoto2008bitcoin}
\bibfield{author}{\bibinfo{person}{Satoshi Nakamoto}.}
  \bibinfo{year}{2008}\natexlab{}.
\newblock \bibinfo{title}{Bitcoin whitepaper}.
\newblock
\newblock


\bibitem[\protect\citeauthoryear{Nguyen, Jourjon, Potop-Butucaru, and
  Thai}{Nguyen et~al\mbox{.}}{2019}]%
        {nguyen2019impact}
\bibfield{author}{\bibinfo{person}{Thanh Son~Lam Nguyen},
  \bibinfo{person}{Guillaume Jourjon}, \bibinfo{person}{Maria Potop-Butucaru},
  {and} \bibinfo{person}{Kim~Loan Thai}.} \bibinfo{year}{2019}\natexlab{}.
\newblock \showarticletitle{Impact of network delays on Hyperledger Fabric}. In
  \bibinfo{booktitle}{\emph{IEEE INFOCOM 2019-IEEE Conference on Computer
  Communications Workshops (INFOCOM WKSHPS)}}. IEEE, \bibinfo{pages}{222--227}.
\newblock


\bibitem[\protect\citeauthoryear{Ozisik, Andresen, Levine, Tapp, Bissias, and
  Katkuri}{Ozisik et~al\mbox{.}}{2019}]%
        {graphene}
\bibfield{author}{\bibinfo{person}{A~Pinar Ozisik}, \bibinfo{person}{Gavin
  Andresen}, \bibinfo{person}{Brian~N Levine}, \bibinfo{person}{Darren Tapp},
  \bibinfo{person}{George Bissias}, {and} \bibinfo{person}{Sunny Katkuri}.}
  \bibinfo{year}{2019}\natexlab{}.
\newblock \showarticletitle{Graphene: efficient interactive set reconciliation
  applied to blockchain propagation}. In \bibinfo{booktitle}{\emph{Proceedings
  of the ACM Special Interest Group on Data Communication}}.
  \bibinfo{pages}{303--317}.
\newblock


\bibitem[\protect\citeauthoryear{Spiegelman, Rinberg, and Malkhi}{Spiegelman
  et~al\mbox{.}}{2020}]%
        {spiegelman2019ace}
\bibfield{author}{\bibinfo{person}{Alexander Spiegelman}, \bibinfo{person}{Arik
  Rinberg}, {and} \bibinfo{person}{Dahlia Malkhi}.}
  \bibinfo{year}{2020}\natexlab{}.
\newblock \showarticletitle{ACE: Abstract Consensus Encapsulation for Liveness
  Boosting of State Machine Replication}. In
  \bibinfo{booktitle}{\emph{OPODIS}}.
\newblock


\bibitem[\protect\citeauthoryear{Stathakopoulou, David, and
  Vukolic}{Stathakopoulou et~al\mbox{.}}{2019}]%
        {DBLP:journals/corr/abs-1906-05552}
\bibfield{author}{\bibinfo{person}{Chrysoula Stathakopoulou},
  \bibinfo{person}{Tudor David}, {and} \bibinfo{person}{Marko Vukolic}.}
  \bibinfo{year}{2019}\natexlab{}.
\newblock \showarticletitle{Mir-BFT: High-Throughput {BFT} for Blockchains}.
\newblock \bibinfo{journal}{\emph{CoRR}}  \bibinfo{volume}{abs/1906.05552}
  (\bibinfo{year}{2019}).
\newblock
\showeprint[arxiv]{1906.05552}
\urldef\tempurl%
\url{http://arxiv.org/abs/1906.05552}
\showURL{%
\tempurl}


\bibitem[\protect\citeauthoryear{Yin, Malkhi, Reiter, Golan{-}Gueta, and
  Abraham}{Yin et~al\mbox{.}}{2019}]%
        {DBLP:conf/podc/YinMRGA19}
\bibfield{author}{\bibinfo{person}{Maofan Yin}, \bibinfo{person}{Dahlia
  Malkhi}, \bibinfo{person}{Michael~K. Reiter}, \bibinfo{person}{Guy
  Golan{-}Gueta}, {and} \bibinfo{person}{Ittai Abraham}.}
  \bibinfo{year}{2019}\natexlab{}.
\newblock \showarticletitle{HotStuff: {BFT} Consensus with Linearity and
  Responsiveness}. In \bibinfo{booktitle}{\emph{Proceedings of the 2019 {ACM}
  Symposium on Principles of Distributed Computing, {PODC} 2019, Toronto, ON,
  Canada, July 29 - August 2, 2019}}, \bibfield{editor}{\bibinfo{person}{Peter
  Robinson} {and} \bibinfo{person}{Faith Ellen}} (Eds.).
  \bibinfo{publisher}{{ACM}}, \bibinfo{pages}{347--356}.
\newblock


\bibitem[\protect\citeauthoryear{Zamani, Movahedi, and Raykova}{Zamani
  et~al\mbox{.}}{2018}]%
        {DBLP:conf/ccs/ZamaniM018}
\bibfield{author}{\bibinfo{person}{Mahdi Zamani}, \bibinfo{person}{Mahnush
  Movahedi}, {and} \bibinfo{person}{Mariana Raykova}.}
  \bibinfo{year}{2018}\natexlab{}.
\newblock \showarticletitle{RapidChain: Scaling Blockchain via Full Sharding}.
  In \bibinfo{booktitle}{\emph{Proceedings of the 2018 {ACM} {SIGSAC}
  Conference on Computer and Communications Security, {CCS} 2018, Toronto, ON,
  Canada, October 15-19, 2018}}, \bibfield{editor}{\bibinfo{person}{David Lie},
  \bibinfo{person}{Mohammad Mannan}, \bibinfo{person}{Michael Backes}, {and}
  \bibinfo{person}{XiaoFeng Wang}} (Eds.). \bibinfo{publisher}{{ACM}},
  \bibinfo{pages}{931--948}.
\newblock


\bibitem[\protect\citeauthoryear{Zhang, Gao, Wu, Yan, Wu, Li, and Chen}{Zhang
  et~al\mbox{.}}{2019}]%
        {DBLP:journals/corr/abs-1912-05241}
\bibfield{author}{\bibinfo{person}{Jiashuo Zhang}, \bibinfo{person}{Jianbo
  Gao}, \bibinfo{person}{Zhenhao Wu}, \bibinfo{person}{Wentian Yan},
  \bibinfo{person}{Qize Wu}, \bibinfo{person}{Qingshan Li}, {and}
  \bibinfo{person}{Zhong Chen}.} \bibinfo{year}{2019}\natexlab{}.
\newblock \showarticletitle{Performance Analysis of the Libra Blockchain: An
  Experimental Study}.
\newblock \bibinfo{journal}{\emph{CoRR}}  \bibinfo{volume}{abs/1912.05241}
  (\bibinfo{year}{2019}).
\newblock


\end{thebibliography}

\appendix
\ifdefined\cameraReady
\section{Security Analysis}\label{sec:proofs}

\subsection{DAG}
\label{dagproofs}

\begin{lemma}

The DAG protocol satisfies Integrity. 

\end{lemma}

\begin{proof}

The lemma follows from the assumption on no hash collusion. That
is, it is impossible to find two blocks that are associated with the
same digest.

\end{proof}

\begin{lemma}

The DAG protocol satisfies Block-Availability.

\end{lemma}

\begin{proof}

Validators locally store every block they sign.
A $write(d,b)$ operation completer $b$ is certified.
Since $2f+1$ signatures are required for a block to be certified, we
get that at least $f+1$ honest validators store the block $b$
associated with the digest $d$. A $read(d)$ operation query all
validators and wait for $n-f$ to reply.
Therefore, any $write(d,b)$ operation invoked after $write(d,b)$
completes will get a reply from at least 1 honest party that stores $b$
and the lemma follows.

\end{proof}







\begin{lemma}

The DAG protocol satisfies 
$1/2$-Censorship-Resistance.

\end{lemma}

\begin{proof}

Since block refers to at least $2f+1$ from the previous round, $B$
contains at least $2f+1$ blocks in each round.
The lemma follows since at most $f$ of them were written by byzantine
validators.

\end{proof}

\begin{lemma}

The DAG protocol satisfies $2/3$-Causality.

\end{lemma}

\begin{proof}

There are at most $3f+1$ written blocks associated with each round. Let
$r$ be the round in which $b$ was certified, the lemma follows
immediately from the fact that $B$ contains at most $2f+1$ blocks from
rounds smaller than $r$.

\end{proof}

\begin{lemma}

The DAG protocol satisfies Containment.

\end{lemma}

\begin{proof}

Every block contains its author and round
number and honest validators do not sign two different blocks in the
same round from the same author.
Therefore, since $2f+1$ signatures are required to certify a block, two
blocks in the same round from the same author can never be certified.
Thus, for every certified block that honest validators locally store,
they always agree on the set of digest in the block (references to
blocks from the previous round).
The lemma follows by recursively applying the above argument starting
from the block $b'$. 

\end{proof}

\subsection{Asynchronous consensus}
\label{consproof}

\para{Safety}

\quorum*

\begin{proof}

An honest validator commits a block $b$ in a instance $i$ only if there
are at least $f+1$ nodes in the second round of the instance with links
to $b$. 
Since every block has at least $2f+1$ links to blocks in the
previous round, we get by quorum intersection that every block in the
first round of instance $i+1$ has a path to $b$. 
Therefore, with a simple inductive argument we 
can show that that every block in every round in instances higher than
$i$ have paths to $b$. The lemma follows.

\end{proof}

\begin{lemma}
\label{lem:safetyayx}

Let $b$ and $b'$ be the block leaders of instances $i$ and $i'$,
respectively.
If an honest validator $v$ commits $b$ before $b'$, then no honest
validator commits $b'$ without first committing $b$.

\end{lemma}

\begin{proof}

Since $v$ commits $b$ before $b'$, then there is no path between $b$ to
$b'$ in the DAG. 
Assume by a way of contradiction that some honest validator $v'$
commits $b'$ before $b$. 
Thus, there is no path between $b$ to
$b'$ in the DAG.
However, by Lemma~\ref{lem:quorum}, one of the paths must exists.
A contradiction.

\end{proof}

Lemma~\ref{lem:safetyayx} immediately implies the following:

\safety*

\para{Liveness}

\livenessaux*

\begin{proof}

Consider any set $S$ of $2f+1$ blocks in the second round of instance
$i$.
The total number of links they have to the first round is
$(2f+1)(2f+1) = 4f^2+4f+1$.
The number of possible blocks in the first round of the instance is
$3f+1$.
Therefore, even if every block in the first round has $f$ links from
blocks in $S$, there are still $4f^2+4f+1 - f(3f+1) = f^2+3f+1$ links.
The maximum number of links from blocks in $S$ to each block in
the first round is $2f+1$.
Thus, there are at least $\frac{f^2+3f+1}{2f+1-f} \geq f+1$ blocks in
the first round such that each one of them has at least $f+1$ links
from block in $B$.

\end{proof}

\liveness*

\begin{proof}

Consider any instance $i$.
By Lemma~\ref{lem:livenessaux}, there are at least $f+1$ block in
instance $i$ that satisfy the commit rule. 
Since the adversary do not know the outcome of the coin before these
leaders are determined and since the coin is uniformly distributed, we
get that the probability to elect a leader that satisfies the commit
rule in instance $i$ is at least $1/3$.
Therefore, in expectation, a block leader is committed every 3
instances.
Every instance consists of 3 rounds, but since we combine the last round
of an instance with the first of the next one, we get that, in
expectation, \consname  commits a block leader
every 7 rounds in the DAG.

\end{proof}

\livenessran*

\begin{proof}

Consider an instance $i$, let $b_i$ be the leader of instance $i$ and
let $S$ be a set of $2f+1$ blocks in the second round of instance
$i$. 
Messages delays are distributed uniformly at random, and
each block in $S$ includes references to blocks in the first round of
instance $i$ independently of other blocks in $S$.
Therefore, each block in $S$ includes a reference to $b_i$ with
probability of at least $\frac{2f+1}{3f+1} \geq 2/3$.
We next show that the probability that at least $f+1$ nodes in $S$ include
$b_i$ is $0.74$.
To this end, we compute the probability for $f=1$ as for larger $f$ the probability is higher.
The probability that at least $2$ nodes out of the 3 nodes in $S$ include
$b_i$ is $\frac{8}{27} + \frac{12}{27} = 0.74$.
Thus, for every instance, the probability to elect a leader that
satisfies the commit rule is at least $0.74$.
Thus, \consname commits, in
expectation, a block leader every 4 rounds (or less for $f>1$) in the DAG. 

\end{proof}

%
%

%

\section{Artifact Appendix} 
\subsection{Abstract}
We open-source the Rust implementation of \sysname and \consname, HS-over-\sysname as well as all Amazon Web Services orchestration scripts, benchmarking scripts, and measurements data to enable reproducible results.

\subsection{Description \& Requirements}
\subsubsection{How to access}
Our implementation of \sysname and \consname is hosted on the following public GitHub repository: 
$$\texttt{https://github.com/asonnino/narwhal}$$ 
The graphs in this paper are generated using the following commit:
$$\texttt{70dc862db090260dc77acdaa807b4584475bafc2}$$ 

Our implementation of HS-over-\sysname is hosted on the following public GitHub repository: 
$$\texttt{https://github.com/asonnino/narwhal/tree/narwhal-hs}$$
The graphs in this paper are generated using the following commit:
$$\texttt{add64984c62b7cf24dced6d40df52ae57032cfa3}$$

\subsubsection{Hardware dependencies}
\sysname does not require any particular hardware dependency. We however run all benchmarks on  Amazon Web Services (AWS), using \texttt{m5.8xlarge} instances. They provide 10Gbps of bandwidth, 32 virtual CPUs (16 physical core) on a 2.5GHz, Intel Xeon Platinum 8175, and 128GB memory. For persistent storage, we equip every instance with an Ebs drive \textit{gp2} of 200GB.

\subsubsection{Software dependencies}
Every instance runs a fresh install of Ubuntu Server 20.04. We used the following AMI image: `Canonical, Ubuntu, 20.04 LTS, amd64 focal image build on 2020-10-26'. Our implementation of \sysname requires Rust 1.51+ and the following dependencies (installed through \texttt{apt-get}):
\begin{itemize}
    \item \texttt{build-essential}
    \item \texttt{cmake}
    \item \texttt{clang}
\end{itemize}

\subsubsection{Benchmarks} 
The benchmarks described in this paper require no input data. We however open-source the raw data used to produce the graphs of \Cref{sec:evaluation}: \url{https://github.com/asonnino/narwhal/tree/master/benchmark/data/paper-data}.

\subsection{Set-up} \label{sec:setup}
To facilitate the experimental set-up, we open-source our AWS orchestration scripts. The scripts are written in Python 3.8 and can be installed as follows:
\begin{verbatim}
git clone https://github.com/asonnino/narwhal.git
cd narwhal/benchmark
pip install -r requirements.txt
\end{verbatim}
A detailed tutorial to used the scripts is availlabble on GitHub: \url{https://github.com/asonnino/narwhal/tree/master/benchmark}.
In particular, steps 1-4 of section `AWS Benchmarks' set up an experimental testbed on AWS.

\subsection{Evaluation workflow}
\subsubsection{Major Claims}
Our evaluation (\Cref{sec:evaluation}) demonstrates the following major claims (Cx):
\begin{itemize}
\item (C1) \sysname as a \mempool has advantages over the existing simple \mempool as well as straightforward extensions of it 
\item (C2) The scale-out is effective, in that it increases throughput linearly as expected. 
\item (C3) \consname is a highly performing consensus protocol that leverages \sysname to maintain high throughput when increasing the number of validators (proving our claim that message complexity is not that important).
\item (C4) \sysname provides robustness when some parts of the system inevitably crash-fail or suffer attacks.
\end{itemize}

\subsubsection{Experiments}
We run the following experiments.~\\

\textit{Experiment (E1): [Common Case]: Evaluation of the common-case (no faulty validators) with various committee sizes.}\\\\
\textit{[Preparation]}
Follow step 5 of the tutorial linked in \Cref{sec:setup}. In particular, set the following variable in \texttt{fabfile.py}:
\begin{verbatim}
bench_params = {
    'nodes': [10, 20, 50],
    'workers: 1,
    'collocate': True,
    'rate': [20_000, 50_000, 100_000],
    'tx_size': 512,
    'faults': 0,
    'duration': 300,
    'runs': 2,
}
\end{verbatim}
Adjust in the input rate \texttt{rate} to reproduce the desired data point of \Cref{fig:latency}
~\\

\textit{[Execution]}
Run the following command: \texttt{fab remote}.
~\\

\textit{[Results]}
Follow step 6 of the tutorial linked in \Cref{sec:setup}. In particular, set the following variable in \texttt{fabfile.py}:
\begin{verbatim}
plot_params = {
    'faults': [0],
    'nodes': [10, 20, 50],
    'workers': [1],
    'collocate': True,
    'tx_size': 512,
    'max_latency': [3_500, 4_500]
}
\end{verbatim}
~\\

\textit{Experiment (E2): [Scalability]: Evaluation with many workers on different machines (no faulty validators).}\\\\
\textit{[Preparation]}
Follow step 5 of the tutorial linked in \Cref{sec:setup}. In particular, set the following variable in \texttt{fabfile.py}, where \texttt{X} is the desired number of workers:
\begin{verbatim}
bench_params = {
    'nodes': [4],
    'workers: X,
    'collocate': False,
    'rate': [100_000, 300_000],
    'tx_size': 512,
    'faults': 0,
    'duration': 300,
    'runs': 2,
}
\end{verbatim}
Adjust in the input rate \texttt{rate} to reproduce the desired data point of \Cref{fig:scalability-latency}.
~\\

\textit{[Execution]}
Run the following command: \texttt{fab remote}.
~\\

\textit{[Results]}
Follow step 6 of the tutorial linked in \Cref{sec:setup}. In particular, set the following variable in \texttt{fabfile.py}:
\begin{verbatim}
plot_params = {
    'faults': [0],
    'nodes': [4],
    'workers': [1, 4, 7, 10],
    'collocate': False,
    'tx_size': 512,
    'max_latency': [3_500, 4_500]
}
\end{verbatim}
~\\

\textit{Experiment (E2): [Crash-faults]: Evaluation with 0, 1 and 3 crash-faults.}\\\\
\textit{[Preparation]}
Follow step 5 of the tutorial linked in \Cref{sec:setup}. In particular, set the following variable in \texttt{fabfile.py}, where \texttt{X} is the desired number of crash-faults:
\begin{verbatim}
bench_params = {
    'nodes': [10],
    'workers: 1,
    'collocate': True,
    'rate': [30_000, 70_000],
    'tx_size': 512,
    'faults': X,
    'duration': 300,
    'runs': 2,
}
\end{verbatim}
Adjust in the input rate \texttt{rate} to reproduce the desired data point of \Cref{fig:latency_faults}.
~\\

\textit{[Execution]}
Run the following command: \texttt{fab remote}.
~\\

\textit{[Results]}
Follow step 6 of the tutorial linked in \Cref{sec:setup}. In particular, set the following variable in \texttt{fabfile.py}:
\begin{verbatim}
plot_params = {
    'faults': [0, 1, 3],
    'nodes': [10],
    'workers': [1],
    'collocate': True,
    'tx_size': 512,
    'max_latency': [10_000, 5_000]
}
\end{verbatim}
~\\

\subsection{General Notes}
\label{sec:gnotes}
Please be mindful that the experiments on AWS described above can be (very) expensive.

\fi

\end{document}